\documentclass[12pt]{iopart}
\usepackage{iopams}
\usepackage{graphicx}
\usepackage{bm}
\usepackage{iopams,setstack}
\usepackage{graphicx}
\usepackage{mathrsfs}

\usepackage[usenames]{color}
\definecolor{quique}{rgb}{0.1,0,0.55}

\newcommand{\lag}{\ensuremath{\mathscr{L}}}        

\newcommand{\D}{\displaystyle}
\newcommand{\beq}{\begin{equation}}
\newcommand{\eeq}{\end{equation}}
\newcommand{\beqa}{\begin{eqnarray}}
\newcommand{\eeqa}{\end{eqnarray}}
\def\nn{\nonumber\\}
\def\eq#1{(\ref{#1})}

\def\cd#1{\ensuremath{\nabla_{#1}}}          
\def\pd#1{\ensuremath{\partial_{#1}}}        

\def\dlag#1{\frac{\partial\lag}{\partial\ \! #1}}

\def\st{spacetime}

\def\mch{\scriptscriptstyle}

\def\text#1{{\rm #1}}

\def\lab#1{\label{#1}}

\newtheorem{theorem}{Theorem}[section]
\newtheorem{lem}[theorem]{Lemma}

\newtheorem{cor}[theorem]{Corollary}
\newtheorem{rem}[theorem]{Remark}
\newtheorem{notation}[theorem]{Notation}
\newtheorem{pro}[theorem]{Proposition}

\def\R{\mathbb{R}}
\def\O{\Omega}

\def\a{\alpha}

\def\ve{\varepsilon}

\def\p{\partial}

\def\E{\mathcal{E}}

\begin{document}
\title[An alternative well-posedness property and \dots]{An alternative well-posedness property and static spacetimes with naked singularities}


\author{Ricardo\ E  Gamboa Sarav\'{\i}$^{1,2}$,  Marcela Sanmartino$^3$ and  Philippe Tchamitchian $ ^4$}
\address{$^1$ Departamento de F\'{\i}sica, Facultad de Ciencias
Exactas, Universidad Nacional de La Plata, Casilla de Correo 67, 1900 La Plata, Argentina. }
\address{$^2$ IFLP, CONICET,
Argentina.  }
\address{$^3$ Departamento de Matem\'{a}tica, Facultad de Ciencias
Exactas, Universidad Nacional de La Plata, Casilla de Correo 172, 1900 La Plata, Argentina. }
\address{$^4$ Universit\'e Paul C\'ezanne,CNRS,LATP (UMR 6632),
Facult\'e des Sciences et Techniques, LATP,
Case cour A, Avenue Escadrille Normandie-Niemen, F-13397 Marseille
Cedex 20, France}

\ead{quique@fisica.unlp.edu.ar, tatu@mate.unlp.edu.ar, philippe.tchamitchian@univ-cezanne.fr}

\date{\today}

%
\begin{abstract}
In the first part of this paper, we show that the Cauchy problem for wave propagation in some
static spacetimes presenting a singular time-like boundary is well posed,
if we only require the waves to have finite energy, although no boundary condition is required.
This feature does not come from essential self-adjointness, which is false in these cases, but from a different
 phenomenon that we  call the {\bf alternative well-posedness property}, whose origin  is due to the degeneracy of the metric components near the boundary.

Beyond these examples, in the second part, we  characterize the
type of degeneracy which leads to this phenomenon.
\end{abstract}

\submitto{\CQG}

\section{Introduction}

The aim of this paper  is to investigate the well-posedness of the
Cauchy problem for the propagation of waves in static spacetimes
presenting a singular time-like boundary, and to clarify the
boundary behaviour of the waves. Such singularities can arise when
solving Einstein equations in the vacuum or in situations where
matter is present but located away from singularities.

This line of research has been initiated by Wald in \cite{W}, and
further developed by, among others, the authors of references
\cite{HM,IH,IW, S,SST}.

 The propagation of waves is formally given by a standard wave
equation of the form
$$\partial_{tt} \phi + A \phi = 0,$$
where $A$ is a second-order symmetric operator with variable
coefficients independent of time, defined on $C^\infty_0
(\Omega)$, the space of test functions (see 2.1 for precise
definitions). Hence, the Cauchy problem is well-posed as soon as
one is able to choose a physically meaningful self-adjoint
extension of $A$, on the Hilbert space $H$ naturally associated to
the setting.

In most cases, this is done either by adding suitable boundary
conditions or because $A$ turns out to be essentially
self-adjoint. Since here, there is no physically meaningful
boundary condition, many authors discussed the essentially
self-adjointness of $A$: although never explicitely proved, the
above cited papers are based on the fact that $A$ does not fulfill
this property.

This is the reason why, in \cite{IH}, the authors   suggested to
replace $H$ by another space, on which they claimed that one
recovers the essentially self-adjointness of $A$. However, this is
wrong, and we prove it at the end of the first part.

Other choices have been to deal with the Friedrichs extension of
$A$, the most clear argument for doing so being stated in
\cite{SST}: ``the Friedrichs extension is the only one whose
domain is contained in $H^1$". In this work we show how this
observation turns out to be at the core of the problem:

When no boundary conditions can be provided (physically or
mathematically), by choosing the Friedrichs extension, the
solution turns out naturally to have finite energy. But, of
course, this extension cannot be interpreted as imposing any
boundary condition.

On the other hand, requiring  the solutions the condition to have
finite energy is, in general,  not enough to have a well posed
problem. However, it may occur, that under this sole condition,
the problem turns out to be well-posed without any boundary
condition, i.e., the solution is completely determined by the
Cauchy data. Indeed, as we show in this work, it happens when the
smallest possible space (the completion of smooth compactly
supported functions on $\O$ with the energy norm) coincides with
the largest possible  one: the energy space  (which obviously
contains functions that do not vanish at the boundary).

\bigskip

In part I, we show this for explicit examples: Taub's plane
symmetric spacetime \cite{taub} and its generalization to higher
dimensions \cite{L}, and Schwarzschild solution with negative
mass. The last one has already been discussed in the literature
\cite{HM,IH}.

 The singularity of  spacetime induces non standard
behavior of the coefficients of $A$ when the point reaches the
boundary. Roughly speaking, in the normal direction to the
boundary, the coefficients vanish, while they explode in the
parallel direction. In particular, the normal coefficient vanishes
like some  power of the  distance to the boundary.

 We start by showing that the operator $A$ is not
essentially self-adjoint. Hence, we have to choose one of its many
self-adjoint extensions. Our criterion is to require the waves to
have finite energy at any time.

Contrary to the standard cases (like, for instance, vibrating
strings), where boundary conditions are needed, we show here that
this sole condition is enough to select a unique self-adjoint
extension of $A$. No boundary conditions must be imposed to
functions in the domain of this extension, and the Cauchy problem
is well-posed with initial conditions only.

 However, the absence of boundary condition does not imply the
absence of boundary behaviour. As a matter of fact, it happens
that the waves do have a non trivial boundary behaviour, i.e., a
trace on $\partial \Omega$, which moreover obeys a regularity law
which depends on the above mentioned exponent.

\bigskip
 In the second part, we turn our attention to the general
problem of deciding when boundary conditions are needed or not for
the Cauchy problem to be well-posed in the absence of essential
self-adjointness. Our setting is that for a divergence operator on
the half-space, fulfilling a pointwise ellipticity condition; we
do not prescribe any {\it a priori} form of global ellipticity,
thus allowing arbitrary degeneracies near the boundary (provided
the coefficients remain locally integrable  up to the boundary,
however). Paralleling somehow our first part, we use the Dirichlet
form associated to the operator to define an energy space which is
the largest possible space whose elements all have finite energy.
Then, we say that the operator has the {\bf alternative
well-posedness property} if it has one and only one self-adjoint
extension with domain included in the energy space.

 We clarify the link between this property and the
density in the energy space of test functions compactly supported
in the geometric domain. On the basis of this abstract
preliminary, we finally give a necessary condition and a
sufficient condition for this property to hold. Under a mild
assumption on the coefficients of the operators, which is
reminiscent of Muckenhoupt $A_2$ class, these conditions are
equivalent, so that we obtain a characterization of the {
alternative well-posedness property}. Its application to static
spacetimes such as the examples above is straightforward.

\vspace{1.5cm}

\begin{center}
{{ \bf \large Part I. Wave propagation in Taub's spacetime and
other examples}}
\end{center}

\vspace{.2cm}

\section{  Massless scalar field in Taub's plane symmetric spacetime}\label{massles}

\subsection{Geometric setting }

We consider the $(n+2)$-dimensional \st\ with $n\geq 2$
and line element \beqa ds^2=- z^{-1+\frac1n}\,
\left(dt^2-dz^2\right)+ z^{2/n}\left[(dx^{\mch
1})^2+\dots+(dx^{\mch{n}})^2 \right] ,\nn
-\infty<t<\infty,\quad\quad-\infty<x^{\mch{i}}<\infty,\,1\leq
i\leq n,\quad\quad0<z<\infty \,.\nonumber  \eeqa
 Throughout the paper, we will use the following
notation:  in $\Omega =\R^n\times (0,\infty)$, the
current point is $(x,z)$, with $x=(x^{\mch{1}}, ...,
x^{\mch{n}})\in \R^n$, $z>0$; the Lebesgue measure on $\Omega $ is
$d\mu$, and on $\R^n$ is $dx$; the gradient and the Laplacian on
$\R^n$ are $\nabla= (\pd {x^{\mch{1}}}, ..., \pd {x^{\mch{n}}}
)^T$ and $\Delta= \sum_{i=1}^n\pd {x^{\mch{i}}x^{\mch{i}}}$.

 When $n=2$, this \st\ is Taub's plane symmetric vacuum
solution \cite{taub}, which is the unique nontrivial static and
plane symmetric solution of Einstein vacuum equations
($R_{ab}=0$). It
has a singular boundary at $z=0$ (see 
 \cite{gs} for a detailed study of the properties of this
solution).

 By matching it to inner solutions it
turns out to be the exterior solution of some static and plane
symmetric distributions of matter \cite{gs,gs2,gs3}. In such a
case, the singularities are not the sources of the fields, but
they arise owing to the attraction of distant matter. We
call them {\bf empty repelling singular boundaries} (see also
\cite{gs1}).

 In this spacetime, we consider the propagation of a
massless scalar field with Lagrangian density \beq\label{lag}
\lag= -\frac{1}{2}\nabla^a\phi\, \cd a \phi =
 -\frac{1}{2} g^{ab}\pd a \phi\, \pd b \phi\ , \eeq
 where $\nabla $  denotes the covariant
derivative (Levi-Civita connection).

 As usual, we obtain the field equations by requiring
that the action
$$ S=\int\lag(\cd a\phi,\phi,g_{ab})\sqrt{|g|}\ dt
d\mu\ $$
 be stationary under arbitrary variations of the fields $\delta
\phi$ in the interior of any compact region, but vanishing at its
boundary. Thus, we have
$$ \cd a\left(\dlag{ \cd a \phi}\
\right)=\dlag{ \phi}\ .
$$
In our case, this reads
 \beqa\label{wave}
\cd a \nabla^a \phi\ &=& \frac{\pd a \left(\sqrt{|g|}\,g^{ab}\pd b
\phi \right)} {\sqrt{|g|}}\nn &=& z^{-1+\frac1n}\left(-\,\pd
{tt}\phi +{\frac1z \pd z(z\, \pd z \phi
)}+\frac{1}{z^{1+\frac1n}}\Delta \phi\right)\nn &=&0. \eeqa As it
is well known (see for example \cite{gs4}), from the Lagrangian
density $\lag$ we get the energy-stress tensor
$$
T^{ab}= -\dlag{ \cd a \phi}\ \nabla^b \phi \ + g^{ab} \lag =
g^{ac}g^{bd}\,\pd c \phi\, \pd d \phi -\frac{1}{2}
g^{ab}g^{cd}\,\pd c \phi\,  \pd d \phi\ ,
$$
which is symmetric and, for smooth enough solution of (\ref{wave}),
 has vanishing covariant divergence
($ \cd a T^{ ab}
=0$). Since the \st\ is stationary, $\pd t$ is a Killing vector
field. Therefore by integrating $(-(\pd t)_b T^{bt})$ over the
whole space we get that, whenever it exists, the total energy of
the field configuration is
$$
E(\phi,t)=\frac{1}{2}\int_{\Omega} \left(z\,(\pd t \phi)^2+z\,(\pd
z \phi)^2+\frac{1}{z^{\frac1n}}\,(\nabla \phi)^2\right) \, dx\,
dz.
$$

\subsection{Statement of the results}

To study the properties of the solutions of the wave
equation (\ref{wave}), we start with defining the underlying
elliptic differential operator $A$ and the Hilbert space $H$ on
which $A$ is symmetric.

 Define, when $\varphi \in C_0^{\infty}(\Omega)$, the
operator $A$ by
$$
A\varphi =-\frac{1}{z}\pd z\left(z\, \pd z \varphi
\right)-\frac{1}{z^{1+\frac1n}} \Delta \varphi.
$$
Then consider the Hilbert space
\begin{eqnarray}
H&:=&L^2(\O, z\, d\mu)\nn &=&\{\varphi(x,z) \ :\ \int_{\O}
|\varphi(x,z)|^2\, z\, d\mu <\infty\}.\nonumber
\end{eqnarray}
By construction the operator $A$ is symmetric on $H$, with
\begin{eqnarray}
<A\varphi,  \eta>_H&=&\int_{\O} \left(z \pd z \varphi \, \pd z
\bar{\eta} \,+ \frac{1}{z^{\frac1n}}\nabla \varphi \cdot \nabla
\bar{\eta}\right) d\mu\nn &=&:b(\varphi,\eta)\nonumber
\end{eqnarray}
 for $\varphi, \eta \in C_0^{\infty}(\Omega)$.

 This leads to introducing the ``energy space''
\begin{equation}\label{energyspace}
\mathcal{E}=\{\varphi \in H^1_{loc}(\O)\cap H \ :\,
b(\varphi,\varphi)<\infty\},
\end{equation}
where  $H^1_{loc}(\O)$ is the usual local Sobolev space. It is
straightforward to check that $\E$, equipped with its natural norm
 $$
 \|\varphi\|^2_{\E}:=b(\varphi, \varphi)+\|\varphi\|^2_{H},
$$ is a Hilbert space. This is the largest subspace of $H$ on which
the form $b$ is finite everywhere.
\begin{rem}
The space $C_c^{\infty}(\Omega)$ of the restrictions to $\Omega$
of $C_0^{\infty}(\R^{n+1})$ functions is included in $\E$; to
prove that it is dense in both $H$ and $\E$ is left to the reader.
\end{rem}

Our first question is whether $A$ is essentially
self-adjoint or not: as a result, it is not. However,
we are only looking for those extensions with domain included in
the energy space, because we are interested in waves having finite
energy. When taking into account this restriction, we recover the
uniqueness of the self-adjoint extension of $A$.

\begin{theorem}\label{esencial}
The operator $A$ is not essentially self-adjoint. However, there
exists only one self-adjoint extension of $A$ whose domain $D$ is
included in the energy space $\E$.
\end{theorem}

 We will see later that the domain of this particular
extension is
$$
 D:=\{\varphi \in \mathcal{E} \ : \exists\  C> 0\
 \forall \ \eta \in \mathcal{E}, \ |b(\varphi,\eta)|\leq C \|\eta\|_H \}.
$$

\begin{notation}
For the sake of simplicity, the self-adjoint extension of $A$ given by
the theorem above is denoted the same way.
\end{notation}
In the next section we will show that the uniqueness of such an
extension  comes from the density in $\E$ of
$C_0^{\infty}(\Omega)$. This important density property prevents
the space $\E$ to possess any kind of trace operator or, more
generally, any continuous linear form supported on the boundary.
Hence there is no boundary condition attached to the definition of
$A$.

\medskip
Now, coming back to the wave equation, we take suitable functions
$f$ and  $g$ on $\Omega$ and consider the Cauchy problem
\ \\
$$(P) \left\{\begin{array}{cc}
 \pd {tt}\phi + A\phi &=0,\\
 \phi(0,\cdot)&=f,\\
  \pd t \phi (0,\cdot)&= g.
\end{array}\right.
$$

\begin{theorem}\label{solution}

\begin{enumerate}
\item Assume $f\in \E$  and $g\in H$. Then the problem (P) has a unique solution
$$
\phi \in C([0,\infty) ; \E)\cap C^1([0,\infty) ; H),
$$
and there exists a constant $C>0$ such that

$$
\forall \ t>0\ \ \ \|\phi(t,\cdot)\|_{\E}+ \|\pd t
\phi(t,\cdot)\|_{H}\leq C(\|f \|_{\E}+ \|g\|_{H}).
$$
\item In this case, the energy
$$
E(\phi,t):=\frac{1}{2}\int_{\Omega} \left(z\,(\pd t
\phi)^2+z\,(\pd z \phi)^2+\frac{1}{z^{\frac1n}}\,|\nabla \phi|^2
\right)d\mu
 $$
 is well-defined and conserved:
$$
\forall\  t>0 \ \ \ E(\phi, t)=\frac12 \left( \|g\|_H^2+ b(f,f)
\right).
$$
\item If, in addition, $f\in D$  and $g\in \E$, we have
$$
\phi \in C([0,\infty) ; D)\cap C^1([0,\infty) ; \E),
$$
and, for another constant $C>0$,

$$
\forall\ t>0\ \ \ \|\phi(t,\cdot)\|_{D}+ \|\pd t
\phi(t,\cdot)\|_{\E}\leq C(\|f \|_{D}+ \|g\|_{\E}).
$$
\end{enumerate}
\end{theorem}

This result shows that the Cauchy problem (P) is well posed
without any boundary condition on $\phi$. This does not necessarily mean,
however, that $\phi(t,\cdot)$ vanishes, or has no trace at all,
on the boundary. Indeed, provided $f$ and $g$ are regular enough,
$\phi(t,\cdot)$ does have a trace on $\partial \Omega$ at each time $t>0$, which
is entirely determined by the Cauchy data:

\begin{theorem}\label{trazat1}

Assume $f\in D$ and $g \in \E$, and let $\phi$ be the solution of
(P) given by Theorem \ref{solution}. Then, for each $t>0$, ${\D
\lim_{z\to 0}\phi(t,\cdot,z)}$ exists in $L^2(\mathbb{R}^n)$.

\end{theorem}

By standard arguments, at each fixed $z>0$, the trace
$\phi(t,\cdot,z)$ on the hyperplane $\Gamma_z = : \{(x,z); x \in
\mathbb{R}^n\}$ exists in the Sobolev space ${H}^1(\mathbb{R}^n)$
(even in ${H}^{3/2}(\mathbb{R}^n)$, see for instance \cite{LM}).
This theorem says that the trace of $\phi$ on the boundary exists
as the strong limit of $\phi(t,\cdot,z)$, for the $L^2$-topology,
of its traces on $\Gamma_z$ when $z \rightarrow 0$.

 We could have written down the theorem above with stronger topologies,
namely that of ${H}^1$ and even ${H}^{3/2}$. But there exists a more
striking result. If we denote
by $\phi(t,\cdot,0)$ the trace of $\phi(t,\cdot)$ on $\partial \Omega$, we have

\begin{theorem}\label{trazat2}
Under the preceding hypotheses, $\phi(t,\cdot,0) \in
{H}^{\frac{2n}{n-1}}(\mathbb{R}^n)$ \- for each $t>0$\-, and we
have more precisely $\phi(\cdot,\cdot,0)\in C([0,\infty) ;
{H}^{\frac{2n}{n-1}}(\mathbb{R}^n))$.
\end{theorem}

If it was a classical case, the trace of $\phi(t)$ would be at
most in ${H}^{3/2}(\mathbb{R}^n)$ because $\phi(t)$ would be in
${H}^{2}(\Omega)$. This is what happens for $\phi(t,\cdot,z)$ at
each $z>0$, which has no reason to be in
${H}^{\frac{2n}{n-1}}(\mathbb{R}^n)$ unless additional assumptions
on $\phi$ are made. What happens here is a compensation phenomenon
between the normal and the tangential degeneracies of the
coefficients of $A$ near the boundary, so that there is a gain of
regularity on the trace $\phi(t,\cdot,0)$ with respect to the
regularity of $\phi(t,\cdot,z)$, $z>0$.

 The theorems stated above will be deduced from the
results of the following section.

\section{The domain of $A$ and its properties }

\subsection{The space $\E$ and the definition of $A$}

\begin{lem}\label{densidad}
  $ C_0^{\infty}(\Omega)$ is dense in $\mathcal{E}$.
\end{lem}

 We pointed out in section \ref{massles} that
$C_c^{\infty}(\Omega)$ is included and dense in $\E$. This Lemma,
hence, implies that {\bf{there is no trace operator in $\E$}} .
Let us be more precise. We claim that there exists no topological
linear space $F$ and no operator $T$ such that
\ \\
i) $T$ is linear and continuous from $\E$ to $F$.
\ \\
ii) If $\varphi \in C_c^{\infty}(\Omega)$, then $\varphi(\cdot ,
0) \in F$, and $T \varphi = \varphi(\cdot , 0)$.

 Indeed, property ii) implies that $T$ vanishes on
$C_0^{\infty}(\Omega)$ but not on $\E$, which contradicts   i) and
Lemma \ref{densidad}.

 Let us now prove the Lemma. It suffices to approximate
in $\E$ any $\varphi \in C_c^{\infty}(\Omega)$ by functions in
$C_0^{\infty}(\Omega)$. Let $\varphi \in C_c^{\infty}(\Omega)$. We
first construct $\varphi_{\varepsilon} \in C_c(\Omega) \cap \E$
for all $\varepsilon >0$ such that $\varphi_{\varepsilon} (x,0)
=0$ and $\displaystyle{\lim_{\varepsilon\to
0}\varphi_{\varepsilon}= \varphi \ \mbox{ in} \ \E}$. To this
purpose, we set
$$
\varphi_{\varepsilon}(x,z)= \varphi(x,z)h_{\varepsilon}(z),
$$
where
$$
h_{\varepsilon}(z)=\left\{\begin{array}{cc}
\frac{1}{\ve} |\log z|^{-\ve} \ &\ \mbox{if}\ \ z\leq z(\ve)\\
 1 \  &\ \mbox{if} \ \ z> z(\ve),
\end{array}\right.
$$
with $\displaystyle{|\log z(\ve)|= \ve^{-\frac1{\ve}}}$; then
 $\varphi_{\varepsilon}(x,z)\in C_c(\Omega) \cap \E$ for all $ \ve >0$.
Furthermore, notice that
$$
\int_0^{\infty}| \pd z h_{\varepsilon}|^2z \, dz = \int_{z\leq
z(\ve)}|\log z|^{-2-2\ve}\;\frac{1}{z}\;dz \underset{\ve \to
0}{\longrightarrow}0.
$$
Using this observation and the dominated convergence theorem, it is straightforward to obtain
\begin{equation}\label{a'}
\|\varphi_{\varepsilon}- \varphi \|_{\E}\underset{\ve \to
0}{\longrightarrow}0.
\end{equation}
Then, for all $\varepsilon >0, \alpha>0$, we define
$$
\varphi_{\varepsilon, \alpha}(\cdot, z)=\left\{\begin{array}{cc}
\varphi_{\varepsilon}(\cdot, z-\alpha) &\ \mbox{if}\ z\geq \a,\\
0 &\ \mbox{if} \  z\leq \a.
\end{array}\right.
$$
Since $\varphi_{\varepsilon}$ vanishes on $\partial \Omega$,
we have
$$ \varphi_{\varepsilon, \alpha}
\underset{\alpha \to
0^+}{\longrightarrow}\varphi_{\varepsilon}.$$ With (\ref{a'}) and
the fact that $\varphi_{\varepsilon, \alpha}\in C_0(\Omega) \cap
\E$, this gives the density in $\E$ of $C_0(\Omega) \cap \E$.
The proof of the lemma finishes with a standard
regularization scheme.
\begin{flushright}
$\Box$
\end{flushright}

 We define the operator $\mathcal{A}$ on the domain
$$
D= \{\varphi \in \mathcal{E} \ :\, \exists\  C> 0 \ \forall \ \eta
\in \E, \ |b(\varphi,\eta)|\leq C \|\eta\|_H \}
$$
by the classical procedure: if $\varphi \in D$, there exists $\psi \in H$ such that
$\displaystyle{b(\varphi,\eta)=<\psi\, ,\eta >}$ for all $\eta \in
\E$, and we set $\psi=\mathcal{A} \varphi$.

\medskip

{\bf Proof of Theorem \ref{esencial}:}

By the symmetry of the bilinear form $b(\varphi, \eta)$ and by the definition of
$D$, $(\mathcal{A}, D)$ is a self-adjoint extension of
$(A,C_0^{\infty}(\Omega))$. Now by Lemma \ref{densidad},
$(\mathcal{A}, D)$ is the Friedrichs extension of
$(A,C_0^{\infty}(\Omega))$; thus it is the only extension with
domain included in $\E$.

 In order to see that
$(A,C_0^{\infty}(\Omega))$ is not essentially self-adjoint it is
enough to give a function $\eta$ such that $\eta \in D(A^*)$,
where
$$ D(A^*)= \{\eta \in H \ :\, \exists\ C> 0: \forall
\varphi \in C_0^{\infty}(\Omega), \ |<\eta, A \varphi>|\leq C
\|\varphi\|_H \},
$$
and $\eta \notin \E$.

Taking $\eta$ such that its Fourier transform in $x$ is
$\hat{\eta}(\xi,z)= K_0(\frac{2n}{n-1}\ z^{\frac{n-1}{2n}}|\xi|)$,
where $K_0$ is the modified Bessel function of the second kind, we
have that $\eta \in H, \ A^*\eta=0$,  and $\eta \notin \E$ {(the
properties of the function $\eta$ are discussed below, see Remark
\ref{behavior})},  thus the proof is finished.
\begin{flushright}
$\Box$
\end{flushright}
\medskip

 As already mentioned, from now on we do not distinguish
between both operators, and use the letter $A$ to denote them.

\begin{rem}
Let $\psi$ in $H$. Then, if we consider the equation
$\psi=A\varphi$, where $\varphi\in H$, in the distribution sense,
either it has a unique solution or it has no solution.
Prescribing some boundary condition would lead to an
overdetermined problem, with no solution in general.

\end{rem}

\subsection{Traces of functions in  $D$}

Our second key result is that although there is no trace
operator in the whole space $\E$, every function in the domain $D$
does have a trace on $\partial \Omega$. If $\varphi \in D$ we
denote by $\varphi(z)$ its trace on the subspace $\Gamma_z=
\R^n\times \{z\}$, $z>0$. Such a trace exists thanks to classical
results on elliptic operators with $C^{\infty}$ coefficients,
\cite{LM}.

\medskip
 In order to study the behavior of $\varphi(z)$ in
$L^2(\R^n)$, we consider $\psi \in H$ such that $A\varphi= \psi$.
Taking the Fourier Transform in $x$, we write this equation as
\beqa \label{fourier} \hat{A \varphi}(\xi,z) =-\frac{1}{z}\;\pd
z\big(z\, \pd z \hat{\varphi}(\xi,z)
\big)+\frac{|\xi|^2}{z^{1+\frac1n}}\hat{\varphi}(\xi,z)\,=\hat{\psi}(\xi,z).
\eeqa

Any solution of the differential equation
\begin{equation}\label{ordinaria}
z\, u''(z)+u'(z)-\frac{1}{z^{\frac1n}}u(z)\,=-z\, v(z)
\end{equation}
can be written as
$$u(z)=\left(\int_z^{\infty}s\,v(s)\, u_1(s)\,
ds+C_0\right)\, u_0(z)+
 \left( \int_0^z s\, v(s)\, u_0(s)\, ds+ C_1\right)\, u_1(z),
 $$
with $u_0(z)=I_0\left(\frac{z^{\alpha_n}}{\alpha_n}\right)$ and
$u_1(z)=\frac{1}{\alpha_n}K_0\left(\frac{z^{\alpha_n}}{\alpha_n}\right)$,
where $\alpha_n=\frac{1}{2}-\frac{1}{2n}$ and $I_0$ and $K_0$ are
modified Bessel functions and solutions of the homogeneous
ordinary differential equation associated with (\ref{ordinaria}).

\begin{rem}\label{behavior}
Near the origin, the behavior of the solutions $u_0(z)$ and
$u_1(z)$ is:

$$
\hskip-1,5cm u_0(z)=
1+\frac{z^{{1}-\frac{1}{n}}}{(1-\frac{1}{n})^2}\,
+O\left(z^{{2}-\frac{2}{n}}\right),\;\; u_0'(z)=
\frac{z^{-\frac{1}{n}}}{(1-\frac{1}{n})}\,
+O\left(z^{{1}-\frac{2}{n}}\right),
$$
$$
u_1(z)=-\log (z)+ c + O\left(z^{{1}-\frac{1}{n}}\right),
u_1'(z)=-\frac{1}{z}+c\ {z^{-\frac{1}{n}}}+O\left(\log
(z)z^{-\frac{1}{n}}\right),
$$
where $c$ is a different constant in each equation.

Also, for large values of $z$, we have
\begin{eqnarray} u_0(z)&=&\left({\sqrt{\frac{\alpha_n}{2\pi}}}\,\,z^{-\frac{\alpha_n}2}
+O\left(z^{-\frac{3\alpha_n}2}\right)\right)e^{\frac{z^{\alpha_n}}{\alpha_n}}\nn
u_0'(z)&=&\left({\sqrt{\frac{\alpha_n}{2
\pi}}}\,\,z^{-1+\frac{\alpha_n}2}+O\left(z^{-1-\frac{\alpha_n}2}\right)\right)e^{\frac{z^{\alpha_n}}{\alpha_n}}\nn
u_1(z)&=& \left(\sqrt{\frac{
\pi}{2\,\alpha_n}}\,z^{-\frac{\alpha_n}2}+O\left(z^{-\frac{3\alpha_n}2}\right)\right)e^{-\frac{z^{\alpha_n}}{\alpha_n}}\nn
   u_1'(z)&=& -\left(\sqrt{\frac{
\pi}{2\,\alpha_n}}\,z^{-1+\frac{\alpha_n}2}+O\left(z^{-1-\frac{\alpha_n}2}\right)\right)e^{-\frac{z^{\alpha_n}}{\alpha_n}}.\nonumber
\end{eqnarray}
\end{rem}

 As a consequence of this remark we obtain that for $\psi
\in L^2((0,\infty), z\, dz)$, the only solution $u$ to
(\ref{ordinaria}) such that $u$ and  $\partial_z u$ belong to
$L^2((0,\infty), z\, dz)$ is
$$
u(z)=\left(\int_z^{\infty}s\,v(s)\, u_1(s)\, ds \right)\,
u_0(z)+ \left(\int_0^z s\, v(s)\,u_0(s)\, ds\right)\, u_1(z).
$$

By scaling in the equation (\ref{ordinaria}), we obtain that the only solution to
(\ref{fourier}) which belongs to $ \hat{\E}=: \{\hat{\varphi}: \varphi \in \E \}$  is
 \beqa\label{fourierz}
 \hat{\varphi}(\xi,z) &=& \left(\int_z^{\infty} s\, \hat{\psi}(\xi,s)\, u_1(s|\xi|^{\frac1{\alpha_n}})
 \, ds\right)\ u_0(z|\xi|^{\frac1{\alpha_n}})\nn
  &+&  \left(\int_0^z s\, \hat{\psi}(\xi,s)\, u_0(s |\xi|^{\frac1{\alpha_n}})\, ds
  \right)u_1(z|\xi|^{\frac1{\alpha_n}}).
 \eeqa
 \\
\begin{lem}\label{pz}
Assume $\varphi \in D$ and Supp $\varphi\subset  \R^n \times [0,1]
$. There exists a constant $C$ independent of $\varphi$ such that
$$
 \sup_{z \leq 1}\|z^{\frac1n}\, \partial_z \varphi(z)\|_{L^2(\R^n)}\leq \, C \|A \varphi\|_{H}.
$$
\end{lem}
\begin{proof}
By derivating (\ref{fourierz}) with respect to $z$, we obtain
\begin{eqnarray}
 \partial_{z}\hat{\varphi}(\xi,z)&=&\left(\int_{z}^{\infty}\!\!\! s\, \hat{\psi}(\xi,s)
 \, u_{1}(s|\xi|^{{\frac1{\alpha_n}}})ds\right) |\xi|^{{\frac1{\alpha_n}}} u'_{0}(z|\xi|^{{\frac1{\alpha_n}}})\nn
  &+& \left(\int_{0}^{z}\!\!\! s\, \hat{\psi}(\xi,s)\, u_{0}(s|\xi|^{{\frac1{\alpha_n}}}) ds\right) |\xi|^{{\frac1{\alpha_n}}}
  u'_{1}(z|\xi|^{{\frac1{\alpha_n}}}).
  \end{eqnarray}
In order to estimate this function in $L^2(\R^n)$, we fix $z$ and
consider first the case $ z|\xi|^{\frac1{\alpha_n}}\geq 1 $. We
have
\begin{eqnarray} \left\|\left(\int_{z}^{\infty}\!\! s\,\hat{\psi}(\xi,s)\,
u_{1}(s|\xi|^{{\frac1{\alpha_n}}})ds
\right)|\xi|^{{\frac1{\alpha_n}}}
u'_{0}(z|\xi|^{{\frac1{\alpha_n}}})\
\mathbf{1}_{\{z|\xi|^{\frac1{\alpha_n}}\geq
1\}}\right\|_{L^2(d\xi)}^2\nn
\label{L21}
\fl
\leq \int_{\{z|\xi|^{\frac1{\alpha_n}}\geq 1\}}\left(\int_{z}^{\infty} s\,
|\hat{\psi}(\xi,s)|^2ds\right)\left(\int_{z}^{\infty} s\,
u^2_{1}(s\, |\xi|^{{\frac1{\alpha_n}}})\, ds\right)
 |\xi|^{\frac2{\alpha_n}} u'^2_{0}(z|\xi|^{\frac1{\alpha_n}})\, d\xi.\nonumber
\end{eqnarray}
Using the behavior of $u_1$ and $u_0$ given in the remark
\ref{behavior}, we obtain
\begin{equation}\label{u1}
\int_{z}^{\infty} s\, u^2_{1}(s|\xi|^{\frac1{\alpha_n}})\, ds \leq
C z^{2-2\alpha_n}|\xi|^{-2}
e^{-2\frac{z^{\alpha_n}}{\alpha_n}|\xi|}
\end{equation}
and
\begin{eqnarray}  \label{L21'}
\left\|\left(\int_{z}^{\infty}\!\! s\,\hat{\psi}(\xi,s)\,
u_{1}(s|\xi|^{\frac1{\alpha_n}})ds \right)|\xi|^{\frac1{\alpha_n}}
u'_{0}(z|\xi|^{\frac1{\alpha_n}})\
\mathbf{1}_{\{z|\xi|^{\frac1{\alpha_n}}\geq
1\}}\right\|_{L^2(d\xi)}^2 \leq C\,\|\psi\|^2_H.\nn
\end{eqnarray}

When $ z|\xi|^{\frac1{\alpha_n}}\geq 1 $, we also have
\begin{equation}\label{u0}
\int_{0}^{z} s\, u^2_{0}(s|\xi|^{\frac1{\alpha_n}})\, ds \leq C
z^{2-2\alpha_n}|\xi|^{-2} e^{2\frac{z^{\alpha_n}}{\alpha_n}|\xi|}.
\end{equation}
This gives
\begin{eqnarray}  \label{L22}
\left\|\left(\int_0^{z} s\, \hat{\psi}(\xi,s)\,
u_{0}(s|\xi|^{\frac1{\alpha_n}})\, ds\right)
|\xi|^{\frac1{\alpha_n}} u'_{1}(z|\xi|^{\frac1{\alpha_n}})\;
\mathbf{1}_{\{z|\xi|^{\frac1{\alpha_n}}\geq
1\}}\right\|_{L^2(d\xi)}^2&& \nn \leq
\int_{\{z|\xi|^{\frac1{\alpha_n}}\geq 1\}}\left(\int_0^{z} s\,
|\hat{\psi}(\xi,s)|^2\, ds\right)\left(\int_0^{z} s\,
u^2_{0}(s|\xi|^{\frac1{\alpha_n}})\, ds\right)
 |\xi|^{8} u'^2_{1}(z|\xi|^{\frac1{\alpha_n}})\, d\xi &&\nn
\leq C \int_{\{z|\xi|^{\frac1{\alpha_n}}\geq
1\}}\left(\int_{z}^{\infty} s\, |\hat{\psi}(\xi,s)|^2\, ds\right)
z^{-\alpha_n}|\xi|^{-1}d\xi\, \leq C \|\psi\|^2_H.&&\nn
\end{eqnarray}

In the case $ z|\xi|^{\frac1{\alpha_n}}\leq 1 $, we use Minkowski
inequality and the support condition on $\varphi$ to obtain

\begin{eqnarray}\label{L23}
\left\|\left(\int_{z}^{\infty} s\, \hat{\psi}(\xi,s)\,
u_{1}(s|\xi|^{\frac1{\alpha_n}})\, ds\right)
|\xi|^{\frac1{\alpha_n}} u'_{0}(z|\xi|^{\frac1{\alpha_n}})\
\mathbf{1}_{\{z|\xi|^{\frac1{\alpha_n}}\leq 1\}}\right\|\nn \leq C
z^{-\frac1n} \int_{z}^{\infty} \sqrt{s}\,
\|\hat{\psi}(s)\|_{L^2(d\xi)}\, ds \leq Cz^{-\frac1n}\|\psi\|_H,
\end{eqnarray}
and
\begin{eqnarray}  \label{L24}
\left\|\left(\int_{0}^{z} s\, \hat{\psi}(\xi,s)\,
u_{0}(s|\xi|^{\frac1{\alpha_n}})\, ds\right)
|\xi|^{\frac1{\alpha_n}} u'_{1}(z|\xi|^{\frac1{\alpha_n}})\
\mathbf{1}_{\{z|\xi|^{\frac1{\alpha_n}}\leq
1\}}\right\|_{L^2(d\xi)}\,\nn \leq C \int_{0}^{z} \frac{s}{z}\,\|
\hat{\psi}(\xi,s)\|_{L^2(d\xi)}\, ds \leq  C\, \|\psi \|_H.
\end{eqnarray}

Thanks to (\ref{L21'}), (\ref{L22}), (\ref{L23}) and (\ref{L24}), and recalling that
 $\psi= A\varphi$,
 the Lemma is proved.
\begin{flushright}
$\Box$
\end{flushright}
\end{proof}

\medskip
 Let now $\varphi \in D$, and assume first that Supp
$\varphi\subset  \R^n \times [0,1]$. Then the last lemma applies,
the formula
$$
\varphi(0)= -\int_0^1 \pd z \varphi(z) dz
$$
defines the trace $\varphi(0)$ of $\varphi$ on $\partial \Omega$,
and gives that $\varphi(0) \in L^2(\R^n)$, with
 \begin{equation}\label{t1}
 \|\varphi(0)\|_{L^2(\R^n)}\leq C \|A \varphi\|_H.
\end{equation}

In the general case, where $\varphi$ does not satisfy any support condition, we modify
$\varphi$ in  $\tilde{\varphi}$ by setting
$$
\tilde{\varphi}(x,z)=\varphi(x,z)h(z),
$$
where $h \in C^{\infty}(\R)$ is such that $h(z)=1$ when $z\leq
\frac12$ and $h(z)=0$ when $z\geq 1$. It is easy to check that
$\tilde{\varphi}\in D$ and that
$$
 \|A \tilde{\varphi}\|_{H}\leq \|A \varphi\|_H + C \|\varphi\|_{\E}.
$$

The inequality (\ref{t1}) applied to $\tilde{ \varphi}$, with the fact that
 $\tilde{ \varphi}=\varphi$ on $\R^n \times [0,\frac12] $, ends the proof of
 Theorem \ref{trazat1}, giving the estimate

\begin{equation}\label{t2}
 \|\varphi(0)\|_{L^2(\R^n)}\leq C \left(\|A \varphi\|_H + \|\varphi\|_{\E}\right).
\end{equation}

\medskip
 The additional regularity of the trace is contained in the next
more precise result:

\break

\begin{lem}\label{laplaciano}
Let $\varphi \in D$ and $\psi\in H$ such that $A\varphi=\psi$.
Then
\begin{equation}\label{i}
\hskip-1.5cm i)\ \ \sup_{z \leq \frac12}\|z^{\frac1n}\, \partial_z \varphi(z)\|_{L^2(\R^n)}\leq \,
 C \left(\|A \varphi\|_{H} + \|\varphi\|_{\E}\right),
\end{equation}
\begin{equation} \label{ii}
\hskip-1.5cm ii)\ \ \Delta \varphi(0)\in L^2(\R^n)\ \ and\ \
 \|\Delta \varphi(0)\|_{L^2(\R^n)}
 \leq C \left(\|A \varphi\|_H + \|\varphi\|_{\E}\right),
\end{equation}
\begin{equation} \label{iii}
\hskip-1.5cm iii)\ \ \Delta \varphi(0)=-\lim_{z\to 0} (1-\frac1n)
z^{\frac1n}\, \p_z \varphi(z)\ \ {in\ the\ strong\  topology\ of\
} L^2(\R^n).
\end{equation}
\begin{equation} \label{iv}
\hskip-1.5cm iv)\ \ \Delta ^{\frac{n}{n-1}} \varphi(0)\in L^2(\R^n).
\end{equation}

v) Moreover, the preceding statement is sharp, in the sense
that for any \newline    $ \gamma\in H^{\frac{2n}{n-1}}(\R^n)$,
there exist $\varphi\in D$ such that $\varphi(0)=\gamma$.
\end{lem}
\begin{proof}
The proof of i) goes through replacing $\varphi$ with
$\tilde{\varphi}$ and applying Lemma \ref{pz}, as we did for
deducing (\ref{t2}) from (\ref{t1}).

 We continue the proof using the same trick, and in order
to alleviate the notation, we write $\varphi$ instead of
$\tilde{\varphi}$.  The second step of the proof consists in
showing that
\begin{equation}\label{t3}
\lim_{z\to 0}(1-\frac1n) z^{\frac1n}\, \p_z \varphi(z)=\lambda (z),
\end{equation}
where  $\lambda$ is defined by
$$
\hat{\lambda}(\xi)=|\xi|^{2} \int_{0}^{\infty}\!\!
s\,\hat{\psi}(\xi,s)\, u_{1}(s|\xi|^{\frac1{\alpha_n}})ds,
$$
and the limit is taken in the strong $L^2$-topology.

To prove that $\lambda \in L^2(\R^n)$ is straightforward: since $
z \mapsto z^{2\alpha_n}\ u_1(z) $
 is bounded, we have by the support condition on $\varphi$ that
\begin{equation}\label{t33}
|\hat{\lambda}(\xi)|\leq  \int_{0}^{1}\!\!
s^{1-2\alpha_n}\,|\hat{\psi}(\xi,s)|\, ds,
\end{equation}
which implies
\begin{equation}\label{t33'}
\|\lambda\|_{L^2(\R^n)}\leq C \|\psi\|_H,
\end{equation} by Cauchy–Schwarz inequality.

  Now, by Remark \ref{behavior}, $u'_{0}(z)\sim
\frac{z^{-\frac{1}{n}}}{(1-\frac{1}{n})}\,
+O\left(z^{{1}-\frac{2}{n}}\right)$ when $z \to 0$. Using this in the proof of
Lemma \ref{pz}, we let the reader check that it gives
$$
c_n z^{\frac1n}\, \p_z \hat{\varphi}(\xi,z)= |\xi|^{2}{\bf
1}_{\{z|\xi|^{\frac1{\alpha_n}}\leq 1\}} \int_{z}^{\infty}\!\!
s\,\hat{\psi}(\xi,s)\,u_{1}(s|\xi|^{\frac1{\alpha_n}}) ds +
\rho(\xi,z),
$$
where $c_n = 1-\frac1n$ and
\begin{equation}\label{rho}
\|\rho(\cdot,z)\|_{L^2(d\xi)}\leq C z^{\frac1n} \|\psi\|_H.
\end{equation}

We thus have, by Lemma \ref{pz} and the dominated convergence
theorem, that
$$
\lim_{z\to 0}c_n z^{\frac1n}\, \p_z \varphi(z)=\lambda (z),
$$
in the strong $L^2$-topology.

\medskip

The next step is to identify  $\lambda$ with $-\Delta\varphi(0)$,
assuming for the moment an additional regularity on $\varphi$,
namely that $\Delta\varphi \in D$. In this case, we already know
that $-\Delta\varphi(0) \in L^2(\R^n)$, since it is the trace of
$-\Delta\varphi$ on $\partial \Omega$. We then start from formula
(\ref{fourierz}) and write
\begin{eqnarray}
-\widehat{\Delta \varphi}(\xi,z)&=&\left(\int_z^{\infty} s\,
\hat{\psi}(\xi,s)\, u_1(s|\xi|^{\frac1{\alpha_n}})\, ds\right)\
|\xi|^2u_0(z|\xi|^{\frac1{\alpha_n}})\nn &+&\left(\int_0^z
s\,\hat{\psi}(\xi,s)\, u_0(s|\xi|^{\frac1{\alpha_n}})\, ds
\right)|\xi|^2u_1(z|\xi|^{\frac1{\alpha_n}}).
 \end{eqnarray}

We decompose $\hat{\lambda}+ \widehat{\Delta \varphi}$ into three
terms:
\begin{eqnarray}
\hat{\lambda}(\xi)+ \widehat{\Delta \varphi}(\xi,z)&=&
\left(\int_0^{z} s\, \hat{\psi}(\xi,s)\,
u_1(s|\xi|^{\frac1{\alpha_n}})\, ds\right)\ |\xi|^2\nn
&+&\left(\int_z^{\infty} s\, \hat{\psi}(\xi,s)\,
u_1(s|\xi|^{\frac1{\alpha_n}})\, ds\right)\ |\xi|^2
(1-u_0(z|\xi|^{\frac1{\alpha_n}}))\nn
&-& \left(\int_0^z s\,
\hat{\psi}(\xi,s)\, u_0(s|\xi|^{\frac1{\alpha_n}} )\, ds
\right)|\xi|^2 u_1(z|\xi|{^{\frac1{\alpha_n}}})\nn
&=&:a(\xi,z)+b(\xi,z)+c(\xi,z).
 \end{eqnarray}

That $\|a(\cdot,z)\|_{L^2(\R^n)}$ tends to 0 with $z$ comes from
the analog of (\ref{t33}). Using the estimates (\ref{u1}) and (\ref{u0}) as in the proof of Lemma \ref{pz}, we see that
\begin{eqnarray}
|b(\xi,z)\ {\bf 1}_{\{z|\xi|^{\frac1{\alpha_n}}\geq
1\}}|^2+|c(\xi,z)\ {\bf 1}_{\{z|\xi|^{\frac1{\alpha_n}}\geq
1\}}|^2\nn \leq C \ |\xi|\ z^{5/4}\left(\int_z^{\infty} s\,
|\hat{\psi}(\xi,s)|^2 \, ds + \int_0^z s\, |\hat{\psi}(\xi,s)|^2\,
ds \right).\nonumber
\end{eqnarray}
This gives
\begin{equation}   \label{b+c}
\int_{\{z|\xi|^{\frac1{\alpha_n}}\geq 1\}} \big(|b(\xi,z)|^2 +
|c(\xi,z)|^2\big)\, d\xi \leq z^{5/4}\|(-\Delta)^{1/4}\psi\|_H^2.
\end{equation}
From he behavior of $u_1$ and $u_0$ near the origin (remark 3.3),
by using the Cauchy-Schwarz inequality we get the estimate
\begin{eqnarray}
 \left|b(\xi,z){\bf 1}_{\{z|\xi|^{\frac1{\alpha_n}}\leq 1\}}\right|&\leq &C z^{1-\frac1n}
 \int_z^{\infty} s|\xi|^{4}\, |\hat{\psi}(\xi,s)|\,|u_1(s|\xi|^{\frac1{\alpha_n}})| \,
 ds.\nn
 &\leq &C z^{\frac1n}\left(\int_z^{\infty} s\, |\hat{\psi}(\xi,s)|^2
 ds\right)^{\frac12},\nonumber
  \end{eqnarray}
  since $s \mapsto s u_1^2(s)$ is integrable. Thus we have

\begin{equation}\label{b}
\int_{\{z|\xi|^{\frac1{\alpha_n}}\leq 1\}} |b(\xi,z)|^2\, d\xi
\leq C z^{\frac2n} \|\psi\|_H^2.
\end{equation}

Finally, we have
\begin{eqnarray}
 \left|\,c(\xi,z){\bf 1}_{\{z|\xi|^{\frac1{\alpha_n}}\leq 1\}}\right|\leq
 C\left(\int_0^z s\, |\hat{\psi}(\xi,s)| ds\right)|\xi|^2|\ln(z|\xi|^{\frac1{\alpha_n}})| \,
 \nn
 \leq C|\xi|^{2\alpha_n-\frac12}\int_0^z\sqrt{s}\, |\hat{\psi}(\xi,s)|ds
\leq C z^{\frac1n}\left(\int_0^z\, s|\hat{\psi}(\xi,s)|^2
ds\right)^{\frac12},
\end{eqnarray}
and this gives
\begin{equation} \label{c}
\int_{\{z|\xi|^{\frac1{\alpha_n}}\leq 1\}} |c(\xi,z)|^2\, d\xi
\leq C z^{\frac2n}\|\psi\|_{H}^2.
\end{equation}

We thus have proved (by the comments on $a$, (\ref{b+c}), (\ref{b}),
and (\ref{c})) that
\begin{equation}\label{fin}
\lim_{z\to 0}\|\lambda + \Delta \varphi(z)\|_{L^2(\R^n)}=0.
\end{equation}

We are now ready for the final step in proving ii) and iii).
 Let $\varphi \in D$, with $\mbox{Supp}\ \varphi \subset \R^n
\times [0,1]$ as we said, and for all $\epsilon >0$, let
\begin{eqnarray}
\varphi_{\varepsilon}&=&(I-\varepsilon\Delta)^{-1}\varphi,\nn
\psi_{\varepsilon}&=&A\varphi_{\varepsilon}=(I-\varepsilon \Delta
)^{-1}\psi.\nonumber
\end{eqnarray}
Then $\displaystyle{\|\psi_{\varepsilon}\|_H\leq \|\psi\|_H}$ and
$\displaystyle{\lim_{\varepsilon \to 0}\|\psi_{\varepsilon}-\psi\|_H =0}$,
 $\varphi_{\varepsilon} \in D,\ \Delta\varphi_{\varepsilon} \in D$, and
$\mbox{Supp}\ \varphi_{\varepsilon} \subset \R^n \times [0,1]$.

Apply Lemma \ref{pz}, (\ref{t3}) and (\ref{fin}) to $\varphi_{\varepsilon}$: it gives
$$
\|\Delta\varphi_{\varepsilon}(0)\|_{L^2(\R^n)} \leq
C\|\psi_{\varepsilon}\|_H\leq C \|\psi\|_H,
$$
therefore, we obtain that $\Delta\varphi(0)\in L^2(\R^n)$, with
$$\|\Delta\varphi(0)\|_{L^2(\R^n)} \leq C \|\psi\|_H.$$
This last estimate and Lemma \ref{pz} are finally applied to
$\varphi-\varphi_{\varepsilon}$ in the following chain of
inequalities:
\begin{eqnarray}
\|\Delta \varphi(0) + c_n z^{\frac1n}\, \p_z
\varphi(z)\|_{L^2(\R^n)}\nn \leq
\|\Delta\varphi(0)-\Delta\varphi_{\varepsilon}(0)\|_{L^2(\R^n)} +
\|\Delta \varphi_{\varepsilon}(0) + \frac12 \sqrt{z}\, \p_z
\varphi_{\varepsilon}(z)\|_{L^2(\R^n)}\nn + c_n \|z^{\frac1n}\,
\p_z \varphi(z)-z^{\frac1n}\, \p_z
\varphi_{\varepsilon}(z)\|_{L^2(\R^n)}\nn \leq C
\|\psi-\psi_{\varepsilon}\|_H + \|\Delta \varphi_{\varepsilon}(0)
+ c_n z^{\frac1n}\, \p_z
\varphi_{\varepsilon}(z)\|_{L^2(\R^n)}.\nonumber
\end{eqnarray}

We know that the first term above tends to zero with $\varepsilon$,
and that the second tends to  zero with $z$ for each fixed $\varepsilon$.
This implies that
$$
\lim_{z\to 0} \|\Delta \varphi(0) + c_nz^{\frac1n}\, \p_z
\varphi(z)\|_{L^2(\R^n)}=0,
$$
and the proof of (\ref{ii}) and (\ref{iii}) is complete.

 To prove (\ref{iv}), we define
$$
\hat{\mu}(\xi):=|\xi|^{\frac{1}{n-1}}\hat{\lambda}(\xi)=
\int_{0}^{\infty}\!\!
s\,|\xi|^{\frac1{\alpha_n}}\hat{\psi}(\xi,s)\,
u_{1}(s|\xi|^{4})ds.
$$
Then, by Cauchy-Schwarz inequality and Remark \ref{behavior}
\begin{eqnarray}
|\hat{\mu}(\xi)|^2 &\leq & \left(\int_{0}^{\infty} s\,
|\hat{\psi}(\xi,s)|^2\,ds\right) \left(\int_{0}^{\infty}
 s\,|\xi|^{\frac1{\alpha_n}}\,|u_{1}(s|\xi|^{4})|^2 ds\right)\nn
 &\leq& C \int_{0}^{\infty}\!\! s\,
|\hat{\psi}(\xi,s)|^2\,ds\, ,\nonumber
 \end{eqnarray}
and we obtain
$$
\int_{\R^n}|\hat{\mu}(\xi)|^2 d \xi \leq C \ \|\psi\|^2_H.
$$
This shows that $\mu \in L^2(\R^n)$, {\it i.e.},
$\Delta^{\frac{n}{n-1}}\varphi(0) \in L^2(\R^n)$.

\medskip
 In order to prove v) let $F \in C^{\infty}(\R)$ such
that $F, F', F''$ are bounded, verifying $F(0)=1,
F'(0)=\frac1{c_n^2}$, and
\begin{equation}\label{CF}
\int_0^{\infty}\Big|F(z^{1-\frac1n})-c_n^2
F'(z^{1-\frac1n})-c_n^2z^{1-\frac1n}F''(z^{1-\frac1n})\Big|^2\
\frac{dz}{z^{1+\frac2n}}
\end{equation}
is finite.

 For example, one can take $
 F(t)=(1+\frac1{c_n}t)^{\frac1{c_n}}e^{-t^4}$ .

 Let $h \in C^{\infty}\big([0, \infty)\big)$ such that
$h(z)=1$ if $z \leq 1$ and $h(z)=0$ if $z \geq 2$. We define
$\varphi$ by its Fourier transform in $x, y$:
$$
\hat {\varphi}(\xi,z):=F(z^{1-\frac1n}\,|\xi|^{2})\,
\hat{\gamma}(\xi)\, h(z),
$$
where $\gamma$ is any arbitrary function in
$H^{\frac{2n}{n-1}}(\R^n)$. We will successively prove that $
\varphi\in H $, $ \varphi\in \E$ and $ \varphi\in D $.

$\bullet$  $ \varphi\in H $ :
$$
\int_{\R^{n+1}_+}|\hat {\varphi}(\xi,z)|^2 z\,  d\mu \leq C
\int_{\R^n}|\hat{\gamma}(\xi)|^2 d \xi<\infty
$$
since $h$ is compactly supported and $F$ is bounded on
$[0,\infty]$.

$\bullet$ $\varphi\in\E$: we first compute $\partial_z
{\varphi}(\xi,z) $, obtaining as its Fourier transform
$$
\partial_z \hat {\varphi}(\xi,z)=\frac{c_n}{z^{\frac1n}}\ F'(z^{1-\frac1n}\,|\xi|^{2})
\, |\xi|^{2}\,\hat{\gamma}(\xi)\, h(z) +
F(z^{1-\frac1n}\,|\xi|^{2})\, \hat{\gamma}(\xi)\, h'(z).
$$
It belongs to $H$ because $F'$ is bounded and
$|\xi|^{\frac{n}{n-1}}\,\hat{\gamma}(\xi) \in L^2(\R^n)$; then, we
turn to $\nabla \varphi$ and write
$$
 \int_{\R^{n+1}_+}|\nabla
\varphi(\xi,z)|^2 \frac1{z^{\frac1n}}\, d\mu
=\int_{\R^{n+1}_+}|\xi|^2\,|\hat {\varphi}(\xi,z)|^2 \,
\frac{d\mu}{z^{\frac1n}} \leq C\int_{\R^n}|\xi|^2\
|\hat{\gamma}(\xi)|^2 d \xi < \infty,
$$ where again we have used
that $F$ is bounded and $h$ is compactly supported.

$\bullet$ $\varphi \in D$: we start by computing
$\widehat{A\varphi}(\xi,z)$
\begin{eqnarray}
\widehat{A\varphi}(\xi,z)&=&\frac{\hat{\gamma}(\xi) \,
h(z)}{z}\Big\{\frac1{z^{\frac1n}}|\xi|^2F(z^{1-\frac1n}\,|\xi|^{2})
-\frac{c_n^2}{z^{\frac1n}}|\xi|^{2}F'(z^{1-\frac1n}\,|\xi|^{2})\nn
&-&\frac{c_n^2}{z^{\frac2n}}|\xi|^{4}F''(z^{1-\frac1n}\,|\xi|^{2})\Big\}
-2c_n\frac1{z^{\frac1n}}F'(z^{1-\frac1n}\,|\xi|^{2})|\xi|^{2}
\hat{\gamma}(\xi) h'(z)\nn
 &-& \frac{1}{z}F(z^{1-\frac1n}\,|\xi|^{2})\hat{\gamma}(\xi) h'(z)
-F(z^{1-\frac1n}\,|\xi|^{2})\hat{\gamma}(\xi) h''(z)\nn
 & = & :\frac{1}{z^{1+\frac1n}}G(z^{1-\frac1n}\,|\xi|^{2})\
 \hat{\gamma}(\xi)\ h(z) + \rho(\xi,z).\nonumber
\end{eqnarray}

It is straightforward to see that  $\rho \in H$. On the other hand, we have
$$
\int_0^{\infty}\int_{\R^n}
\Big|\frac{1}{z^{1+\frac1n}}G(z^{1-\frac1n}\,|\xi|^{2})\Big|^2
|\hat{\gamma}(\xi)|^2h^2(z) z d\mu\leq
\int_{\R^n}|\xi|^{\frac{4n}{n-1}}\,|\hat{\gamma}(\xi)|^2 d\xi
<\infty,
$$
where in the last step we have made use of \eq{CF}, whence
$A\varphi \in H$.

We thus have  constructed a function $\varphi$ in the
domain $D$ of $A$ such that $\varphi(0)=\gamma$ (since
$F(0)=h(0)=1$). Note that $\left.c_n
z^{\frac1n}\partial_z\varphi(z)\right|_{z=0}=- \Delta \varphi(0)$,
as required. The proof of Lemma \ref{laplaciano} is complete.
\begin{flushright}
$\Box$
\end{flushright}
\end{proof}

\section{Solutions of the wave equation. Existence and properties.}

 In this section we apply the results of our study of the
operator $A$ to the resolution of the Cauchy problem for the wave
equation.

\subsection{Well-posedness of (P)}

Here we prove the assertions {\it i}) and {\it iii}) of
Theorem \ref{solution}.

 Given $f \in D$ and $g\in \E$, the solution of (P) is
given by (see, for example, \cite{K})

\begin{equation}\label{onda}
\phi(t,\cdot)= \cos(tA^{\frac12})f+
A^{-\frac12}\sin(tA^{\frac12})g.
\end{equation}
Taking into account that $D(A^{\frac12})= \E$, we have
$\phi(t,\cdot) \in D$ and $\pd t\phi(t,\cdot)  \in \E$. That
$\phi(t,\cdot) $ and $\pd t\phi(t,\cdot) $ are continuous
vector-valued functions (in $D$ and in $\E$ respectively) rely on
a classical density argument we only sketch. For $\varepsilon >0$
we set $f_{\varepsilon}=(I+\varepsilon A)^{-1}f$,
$g_{\varepsilon}=(I+\varepsilon A)^{-1}g$ and
$\phi_{\varepsilon}=(I+\varepsilon A)^{-1}\phi$. Then $\pd
t\phi_{\varepsilon}(t,\cdot) \in D$ and $\pd
{tt}\phi_{\varepsilon}(t,\cdot) \in \E$, with their norms
uniformly bounded in $t$, while  $\phi_{\varepsilon}(t,\cdot)
\rightarrow \phi(t, \cdot) $ in $ D$  and $\pd
t\phi_{\varepsilon}(t,\cdot)\rightarrow \pd t\phi(t,\cdot)$ in
$\E$ when $\varepsilon \to 0$. The conclusion readily follows.

\medskip
 When $f \in \E$ and $g\in H$, we define $\phi(t,\cdot)$
by (\ref{onda}). Then
 $\phi(t,\cdot) \in \E$ and $\pd
t\phi(t,\cdot)  \in H$. The continuity results are obtained by
density arguments in the same way as above.

 The reader should notice that in this case we have
$\pd{tt}\phi(t,\cdot)+A(\phi(t,\cdot))=0$ in $\E'$, where $\E'$ is
the dual space of $\E$; hence $\phi$ is a weak solution of (P).
\subsection{Conservation of the energy}

 Although the argument here is standard, we recall it for
the convenience of the reader. We assume first that $f \in D$ and
$g\in \E$. Then $\phi(t,\cdot)$ is a strong solution of (P) and we
have
\begin{equation}\label{yanose}
\int_{t_1}^{t_2}\int_{\Omega}z\, \pd t\phi \, (\pd {tt} \phi + A
\phi)\, dt\, d\mu = 0.
\end{equation}

We consider each term separately, obtaining for the first one
\begin{eqnarray}  \label{tiempo}
\int_{\Omega}\int_{t_1}^{t_2}z\, \pd t\phi\ \pd {tt} \phi \, dt\,
d\mu=\left.\frac12\int_{\Omega}z\, ( \pd {t} \phi )^2 \,
d\mu\right|_{t_1}^{t_2},
\end{eqnarray}
and for the second one (see for instance \cite{K})
\begin{eqnarray}  \label{comilla}
\int_{t_1}^{t_2}\int_{\Omega}\pd t\phi\ A \phi\, z\, dt\, d\mu &=&
\int_{t_1}^{t_2}<\pd t\phi , A \phi >_H\, dt\nn
&=& \int_{t_1}^{t_2}b(\phi, \pd t\phi )\, dt\nn &=&
\frac12\left.\int_{\Omega}\left((\pd z \phi )^2 z  + \frac
1{z^{\frac1n}}\ (\nabla \phi)^2\right) d\mu \ \right|_{t1}^{t_2}.
\end{eqnarray}

Now, by (\ref{yanose}), adding (\ref{tiempo}) and (\ref{comilla}),
we have for all $t>0$
\begin{eqnarray}
  E(\phi,t)&=&\frac{1}{2}\int_{\Omega} \left(z\,(\pd t
\phi)^2+z\,(\pd z \phi)^2+\frac{1}{z^{\frac1n}}\,|\nabla \phi|^2
\right)d\mu \nn &=&\frac12 \left( \|g\|_H^2+ b(f,f)
\right)\nonumber.
\end{eqnarray}
Again, by a density argument as in the preceding subsection, this
result remains true when $f\in\E$ and $g\in H$.

\subsection{Flux of the energy}

 What precedes has a consequence on the behaviour of the
flux of energy through the hyperplanes $z=cte$, which we now
describe.

 Let $f\in D$ and $g\in H$. The flux of energy from the
region $\{z>z_0\}$ to its complement $\{z<z_0\}$ is defined as
$$
\frac{d}{dt}\frac{1}{2}\int_{\{z>z_0\}} \left(z\,(\pd t
\phi)^2+z\,(\pd z \phi)^2+\frac{1}{z^{\frac1n}}\,|\nabla \phi|^2
\right)d\mu \, .
$$
 By direct computation, it is equal to
$$
\int_{\Gamma_{z_0}} T^{zt}\, dx=-\int_{\Gamma_{z_0}} z\,\pd t
\phi\,\pd z \phi\, dx \, .
$$
Writing
$$
\pd t \phi(t, z_0)- \pd t \phi(t, 1)=\int_{z_0}^1\pd z \pd
t\phi(t, z_0) dz\, ,
$$
and using that $\pd t \phi \in \E$, we obtain
$$
\|\pd t \phi(t, \cdot)\|_{L^2(\Gamma_{z_0})}\leq C\left(1+|\log
z_0|\right)^{\frac12}.
$$
With (\ref{i}) in Lemma \ref{laplaciano}, this gives
$$ \Big|\int_{\Gamma_{z_0}} T^{zt}\, dx \Big|\, \leq \,  C
z_0^{1-\frac1n}\left(1+|\log z_0|\right)^{\frac12}.
$$

In particular ${\D \lim_{z_0\to 0} \int_{\Gamma_{z_0}} T^{zt}\,
dx=0}$. This means that the wave is completely reflected at the
boundary.

\subsection{Traces}

Finally, the results on the traces of the waves
$\phi(t, \cdot, \cdot)$ for each $t$ stated in Theorem
\ref{trazat1} and in Theorem \ref{trazat2} follow from Lemma
\ref{pz} and Lemma \ref{laplaciano} respectively.

\section{Vertical waves}

Here, for the sake of completness,  we consider smooth solutions
of the wave equation $\Box \phi=0$ independent of the horizontal
coordinates $x$ and $y$. Indeed for this case we will be able to find
an explicit representation of $\phi$ in terms of the Cauchy data.

 In this case, \eq{wave} becomes
\begin{equation}\label{fe} \pd {tt}\phi(t,z)
-\frac{\pd z\bigl(z\, \pd z \phi(t,z) \bigr)}{z}=0 .
\end{equation}

We are looking for the solution $\phi(t,z)$ of this equation, with
 given  initial values $\phi(t,z)\big|_{t=0}=f(z)$
and $\pd t \phi(t,z)\big|_{t=0}=g(z)$ at the Cauchy ``surface"
$t=0$,
 such that $\pd t \phi(t,.), \pd z \phi(t,.) \in
L^2((0,\infty), z dz)$.
\ \\
\begin{pro}\label{vertical}
Under the above mentioned conditions, the solution of the equation
(\ref{fe}), for smooth enough functions $f(z)$ and $g(z)$, is
given by

\beqa\lab{z>t}
 \phi(t,z)&=&\frac{1}{2}
\sqrt{1-\frac{t}{z}}\;
   f(z-t)+\frac{1}{2} \sqrt{1+\frac{t}{z}}\;
   f(z+t)\nn
   &+&\frac{t}{16
   z\sqrt{z}}{
   \int_{z-t}^{z+t} \frac{f(\zeta) }{\sqrt{\zeta}}\,\;\;
   _2F_1\left(\frac{3}{2},\frac{3}{2};2;\frac{t^2
   -(\zeta-z)^2}{4 \zeta z}\right) \, d\zeta}\nn
   &+&\frac{1}{2
   \sqrt{z}}{
   \int_{z-t}^{z+t} g(\zeta) \,\sqrt{\zeta}\,\;\;
   _2F_1\left(\frac{1}{2},\frac{1}{2};1;\frac{t^2
   -(\zeta-z)^2}{4 \zeta z}\right) \, d\zeta}
   \,. \eeqa
for $t\leq z$, and
\beqa\lab{z<t}\fl  \phi(t,z)=
\int_{0}^{t-z} \Bigl( f'(\zeta) +g(\zeta)\Bigr) \frac{\zeta }{\sqrt{t^2
-(\zeta-z)^2}}\;\,
_2F_1\left(\frac{1}{2},\frac{1}{2};1;\frac{4 \zeta z}{t^2
-(\zeta-z)^2}\right) \, d\zeta \nn %
\fl +\,(t-z) \int_{0}^{t-z}f(\zeta)\, \frac{1
}{(t-z+\zeta)\;\sqrt{t^2-(\zeta-z)^2}}\;
_2F_1\left(\frac{1}{2},\frac{1}{2};1;\frac{4 \zeta z}{t^2
-(\zeta-z)^2}\right) \, d\zeta \nn %
\fl + \,z \int_{0}^{t-z}f(\zeta) \frac{\zeta\;
(t-z-\zeta)}{(t-z+\zeta)\,\bigl(t^2
-(\zeta-z)^2\bigr)^{3/2}}\,\frac{}{}\,\;
_2F_1\left(\frac{3}{2},\frac{3}{2};2;\frac{4 \zeta z}{t^2
-(\zeta-z)^2}\right) \, d\zeta \nn
\fl +\frac{1}{2
\sqrt{z}}{
\int_{t-z}^{t+z}  \Bigl(f'(\zeta) +g(\zeta)+\frac{f(\zeta)}{2\zeta}\Bigr) \,\sqrt{\zeta}\,\;\;
_2F_1\left(\frac{1}{2},\frac{1}{2};1;\frac{t^2
-(\zeta-z)^2}{4 \zeta z}\right) \, d\zeta}\nn
\fl +\frac{1}{32 z
\sqrt{z}}{
\int_{t-z}^{t+z}  f(\zeta)  \frac{\bigl( z^2-(t-\zeta)^2\bigr)}{\zeta\sqrt{\zeta}}\,\;\;
_2F_1\left(\frac{3}{2},\frac{3}{2};2;\frac{t^2
-(\zeta-z)^2}{4 \zeta z}\right) \, d\zeta}%
\,. \eeqa
for $t \geq z$, where $\;_2F_1$
is the Gauss hypergeometric function.

In particular, at the boundary, we have the trace
   \beqa \lab{z=0} \phi(t,0)=
   f(0)+\int_{0}^{t}  \frac{\Bigl(t\, f'(\zeta) +\zeta\,g(\zeta)\Bigr) }{\sqrt{t^2
   -\zeta^2}} \, d\zeta
   \,. \eeqa
\end{pro}

  Notice that, as already has  been pointed out, the solution is
completely determined by the initial data, and no boundary
condition should and can be provided.

 For the sake of readability,  the proof of this
proposition is postponed to  \ref{apendice}.

 \subsection{Two explicit examples}

 In order to explore the qualitative behavior of these
waves we have explicitly computed some solutions of the Cauchy
problem \eq{fe}. In this section, we show  two specific examples.
In both cases, we choose two particular $C_0^1(\R_+)$ functions
$f$ and $g$  and get $\phi(t,z)$ by numerically integrating the
expressions given in \eq{z>t} and \eq{z<t}.

\subsubsection{Example 1}

We take, at $t=0$, the Cauchy data
$$ f(z)= \cases{0 & $0\leq
z\leq1$\\16(z-1)^2(z-2)^2 \quad\quad& $1\leq z\leq 2$\\0 & $2\leq
z$ }\quad\quad\quad\mbox{and}\quad\quad\quad g(z)=0\,.
$$

In Fig. \ref{Ej1} we display two views of the plot of $\phi(t,z)$ for
$t\in[0,3]$ and $z\in[0,5]$ obtained from \eq{z>t} and \eq{z<t} by
numerical integration.
 Notice that the initial
pulse is decomposed in two pieces,  as occurs with D'Alembert's
solution. One of the waveforms travels in the positive
z-direction, while the other travels in opposite direction. Then
the latter reaches the boundary and $\phi(t,0)$ increases from $0$
to a maximum value.  Later on, it becomes negative and attains a
minimum, afterward it tends to $0$  as $C/t^2$ as it can be
readily seen from \eq{z=0}. Thus we see how it is reflected at the
singular boundary and proceeds to travel upward.
\begin{figure}[h!]\begin{center}
\includegraphics[width=\textwidth]{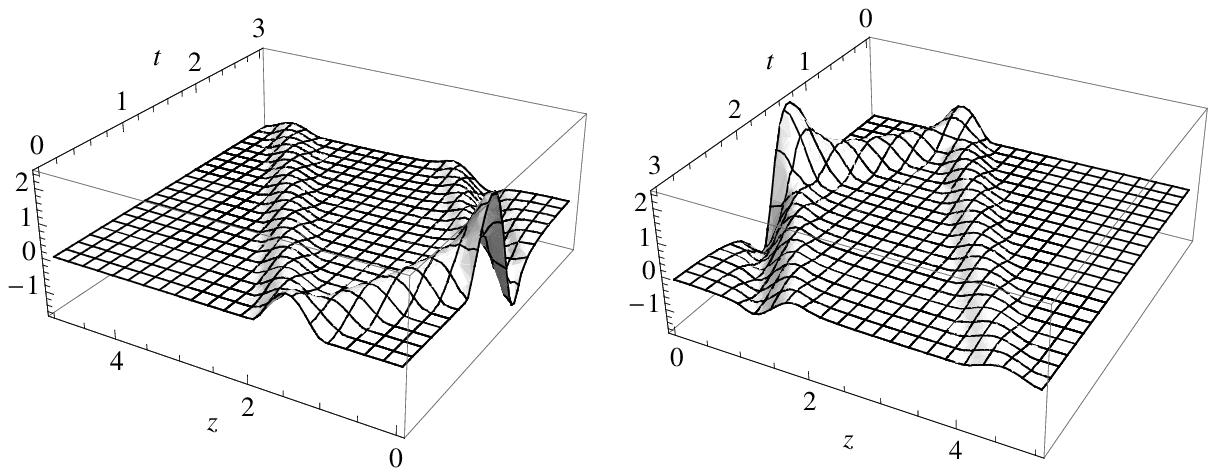}
\caption{\lab{Ej1}Example 1, two views of
$\phi(t,z)$.}\end{center}\end{figure}
\subsubsection{Example 2} In this case, interchanging the roles of
$f$ and $g$, we  set
$$
f(z)=0 \quad\quad\quad\mbox{and}\quad\quad\quad g(z)=\cases{0 &
$0\leq z\leq1$\\16(z-1)^2(z-2)^2 \quad\quad& $1\leq z\leq 2$\\0 &
$2\leq z$ }\,.
$$

In Fig. \ref{Ej2} we display two views of the plot of
$\phi(t,z)$  for $t\in[0,3]$ and $z\in[0,5]$ obtained from
\eq{z>t} and \eq{z<t}.
 In this case, we readily
get from \eq{z=0}  that $\phi(t,0)$ tends to $0$ as $C/t$. We can
see from Fig. \ref{Ej2}, that the wave is also completely
reflected at the boundary.
\begin{figure}[h!]\begin{center}
\includegraphics[width=\textwidth]{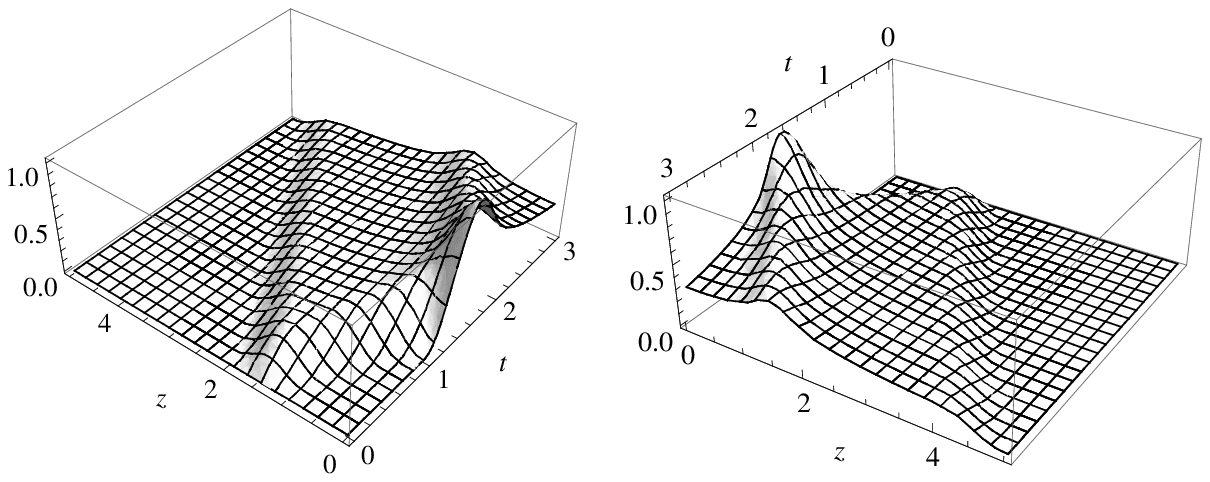}
\caption{\lab{Ej2}Example 2, two views of
$\phi(t,z)$.}\end{center}\end{figure}

 These examples clearly show that, in spite of
$\phi(t,z)$ does not satisfy any prescribed boundary condition at
the origin, some of the qualitative features of the waves are
similar to those of the classical cases, such as vibrating strings
or pressure waves in pipes, with  boundary condition imposed at
the end.  In fact, in all the cases the original pulse decomposes
in two pieces
 and the one travelling to the boundary completely
reflects at it.

However, in the latter cases, the corresponding operator $A$ does
not satisfy the well-posedness property   and  admits infinitely
many self-adjoint extensions with domain included in the energy
space. For instance, the propagation of pressure waves inside the
tube of a wind musical instrument depends drastically on the
physical properties of its end. When it is open, the air pressure
must equal the atmospheric pressure there, and we have to impose
Dirichlet's boundary conditions. Being the end closed, air cannot
move at it, and  Neumann's boundary conditions must be required.
Furthermore, by adding a movable membrane at the extreme, we can
generate more general boundary conditions to be satisfied.
Therefore, in these cases, it is Physics which requires the
existence of infinitely many self-adjoint extensions.

 Whereas, in our case, in which  Physics cannot provide
any boundary condition to be imposed,  the alternative
well-posedness property fortunately tells us that none is actually
needed.

\section{Schwarzschild with negative mass}

The analysis of Schwarzschild \st\ with negative mass
shows similar features than Taub's example, and it is the purpose
of this section to briefly describe it.

We start with
recalling the geometrical setting. The space is $\O=S^2\times
(0,\infty) $ and the metric of the spacetime is

\beqa ds^2=-\left(1 +\frac{2M}{r}\right)
 dt^2+\left(1
+\frac{2M}{r}\right)^{-1}
 dr^2 + r^2 d\Sigma^2 , \nonumber
  \eeqa
where $d\Sigma^2$ denotes the metric on the unit sphere $S^2$ and
$M$ is a positive parameter. In this
case, the wave equation writes
$$
\pd {tt}\phi + A\phi=0
$$
with
 \beqa \label{Asch}A\varphi =-\left(1 +\frac{2M}{r}\right)\frac{1}{r^2}\,
\left[\pd r\left(r^2 \left(1 +\frac{2M}{r}\right) \pd r \varphi
\right)+ \Delta_{S^2} \varphi\right]\nonumber
\eeqa
defined on $C_0^{\infty}(\Omega)$.

 The operator $A$ is
symmetric on the Hilbert space
$$H= \{\varphi(r, \mathbf{\theta}): \int_{\O}|\varphi(r,\mathbf{\theta})|^2
\frac{r^2}{1 +\frac{2M}{r}}dr d\sigma (\mathbf{\theta}\ ,)<\infty
\}
$$ and the associated energy space is defined as in (\ref{energyspace}),
through the sesquilinear form

$$
b(\varphi,\eta)=\int_{\O}\pd r \varphi \ \pd r \eta \
r^2\Big(1+\frac{2M}r\Big) dr d\sigma(\mathbf{\theta})+
\int_{\O}\nabla_{\mathbf{\theta}}\varphi \cdot
\nabla_{\mathbf{\theta}} \, \eta \ r^2 dr d\sigma
(\mathbf{\theta})\ .
$$

\begin{theorem}\label{sch}
The operator $A$ is not essentially self-adjoint. However, it has only
one self-adjoint extension with domain included in $\E$.
\end{theorem}
\begin{proof}
The uniqueness of the self-adjoint extension with domain included
in $\E$ is proved as in Theorem \ref{esencial}. We first show that
$C_0^{\infty}$ is dense in $\E$, just by mimicking the proof of
Lemma \ref{densidad}, and then we usethe same argument to
conclude. The counterexample to essential self-adjointness is
given by any function
$$\varphi(r, \mathbf{\theta})= Q_l\left(1+\frac{r}{M}\right)Y_l^m(\mathbf{\theta})
$$
where $\theta \in S^2$ is the current point in the unit sphere,
$l$ is an integer, $l\geq 1$, $m$ is another integer with $|m|\leq
l$, $Y_l^m$ is the spherical harmonic, and $Q_l$ is the second
Legendre function. We recall that
$Q_l\big(1+\frac{r}{M}\big)=O(|\ln r|)$ and
$Q_l'\big(1+\frac{r}{M}\big)=O(\frac1r)$ when $r\rightarrow 0$,
while $Q_l\big(1+\frac{r}{M}\big)=O(r^{-l-1})$ at infinity. Hence,
$\varphi_{l,m}\in H$, but $\varphi_{l,m}\notin \E$ on one hand, on
the other hand, $A^*\varphi_{l,m}=0$ by construction.\begin{flushright}
$\Box$
\end{flushright}

\end{proof}

 The example $\varphi(r,\mathbf{\theta})=(1+Y_l^m(\mathbf{\theta}))\eta(r)$
shows that functions in the domain of $A$ do not necessarily have
a trace at the origin. This comes from the dependency with respect
to angular variables. However, if we consider for each $\varphi $
in the domain of $A$ its mean values over the spheres $S(0,r)$,
angular variables disappear and, as a result the limit when $r \to
0$ does exists, and is not {\it a priori} vanishing. Indeed, we
have the following.
\begin{theorem}\label{lamda}
Let $\varphi \in D(A)$ and define $$\lambda(\varphi,
r)=\int_{S^2}\varphi(r,\mathbf{\theta})d\sigma(\mathbf{\theta}).$$
Then,
$$\lim_{r\to 0}\lambda(\varphi, r)=: \lambda(\varphi)
$$
exists, and
$$
|\lambda(\varphi)|\leq C (\|A\varphi\|_H +\|\varphi\|_{\E})
$$
for some constant $C$.
\end{theorem}

\begin{proof}
Let $\varphi \in D(A)$, $\psi\in H$ such that  $\psi=A\varphi$;
multiplying it by a cut-off function away from zero, we assume
that Supp\ $\varphi \subset B(0,1)$. We have

\beqa -\left(1 +\frac{2M}{r}\right)\frac{1}{r^2}\, \left[\pd
r\left(r^2 \left(1 +\frac{2M}{r}\right) \pd r \lambda(\varphi, r)
\right)\right]=\int_{S^2}\psi(r,\mathbf{\theta})d\sigma(\mathbf{\theta})\nn
+ \left(1 +\frac{2M}{r}\right)\frac{1}{r^2}\int_{S^2} \Delta_{S^2}
\varphi(r,\mathbf{\mathbf{\theta}}) d\sigma(\mathbf{\theta}))\nn =
\int_{S^2}\psi(r,\mathbf{\theta})d\sigma(\mathbf{\theta}),\nonumber
 \eeqa
from which we deduce that
\begin{eqnarray}\fl
\int_0^1 \left| \pd r \left( r^2 \left(1 +\frac{2M}{r}\right) \pd r
\lambda(\varphi,r)\right)\right|^2\left(1+\frac{2M}{r}\right)\frac{dr}{r^2}&\leq &
\int_{\Omega}|\psi(r,\theta)|^2 r^2\frac{dr d\sigma
(\mathbf{\theta})}{1 +\frac{2M}{r}}\nn & \leq & C \| \psi,, -
\|_H.\nonumber
\end{eqnarray}
Using this estimate and Cauchy-Schwarz inequality, we obtain
for $s<r<1$
\begin{eqnarray}\label{resta}
\left|r^2 \left(1 +\frac{2M}{r}\right) \pd r
\lambda(\varphi,r)-s^2 \left(1 +\frac{2M}{s}\right) \pd s
\lambda(\varphi,s) \right|\nn \leq C \left( \int_s^r \frac{t^2}{1
+\frac{2M}{t}}dt \right)^{\frac12}\| \psi \|_H.
\end{eqnarray}
Then ${\D \lim_{r\to 0} r^2 \left(1 +\frac{2M}{r}\right) \pd r
\lambda(\varphi,r)}$ exists and it must be zero since
$$\int_0^1|\pd r \lambda(\varphi, r)|^2 r dr <\infty
.$$ Thus, when $s$ goes to $0$, we get from (\ref{resta}) that
$$|\pd r \lambda(\phi,r)|\leq C r \left(\int_0^r |\psi(t,\theta)|^2 dt\right)^{\frac12}.$$
Hence
$$
\lambda(\varphi, 0)=-\int_0^1 \pd r \lambda(\varphi, r) dr
$$
exists when $\mbox{Supp}\ \varphi \subset B(0,1)$. The end of the proof
follows as in the proof of (\ref{t2}).
\begin{flushright}
$\Box$
\end{flushright}

\end{proof}

\bigskip
 {\bf Remark.}  Non essential self-adjointness of $A$ has
been claimed in \cite{HM} where a different argument based on von
Neumann's criterion is suggested. However we think preferable to
give a proof.

 Later on, the authors of \cite{IH} suggested to replace
the underlying Hilbert space $H$ by the energy space $\E$. To
avoid confusion, we note $A_{\E}$ their operator, given by
(\ref{Asch}), defined on $C_0^{\infty}$ and taking values in $\E$.
They claimed that, doing this, the operator $A_{\E}$ is
essentially self-adjoint. This is not founded for three reasons.
First, the energy which is associated to such an operator is not
the physical energy. Second, the importance of the density in
$\E$ of $C_0^{\infty}$ has been skipped in \cite{IH}, while it is
a key property which must be verified, otherwise the operator $A$
would not be densely defined (a non densely defined operator
cannot be essentially self-adjoint). Finally, and more
dramatically, their claimed result is just false.

 Indeed, by
Theorem \ref{lamda} there exists a non trivial linear form on
$D(A)$ which vanishes on $C_0^{\infty}$. Hence the closure of the
latter in the former, which we denote by $D_0$, is strictly
included in $D(A)$. One can use this to construct two different
self-extensions of $A_{\E}$. Let $\beta$ be the sesquilinear form
associated to $A_{\E}$ with domain $D(A)$, and $\beta_0$ the
analogous form with domain $D_0$. Then, we classically define
$\mathcal{A}$ and $\mathcal{A}_0$ respectively on:
$$
D(\mathcal{A}) = \{f \in D(A); \, \forall g \in D(A) \, |
\beta(f,g) | \leq C \| g \|_{\E}\}
$$
and
$$
D(\mathcal{A}_0) = \{f \in D_0; \, \forall g \in D_0 \, |
\beta(f,g) | \leq C \| g \|_{\E}\},
$$
where
$$
\beta(f,\psi)=<Af, A\psi>_H .
$$
Both are self-adjoint extensions of $A_{\E}$, and there are
different since their associated forms have distinct domains.\\

\break

\begin{center}
{{\bf \large{Part II. The alternative well-posedness property}}}
\end{center}

\section{Setting of the problem and statement of the main result}

In this second part, we initiate the program which aims at understanding
what is hidden behind the examples of the first part, and to which extent
one can obtain general results.

 We keep on considering $\Omega = \R^n \times
(0,\infty)$, and turn our attention to the divergence operators of
the form
$$\displaystyle{L=-\frac1m \, \mbox{div}\, M\, \mbox{grad}}.$$

We assume:
 \begin{itemize}
 \item (H1)
 $m \in C^\infty (\Omega)$ and $M \in C^\infty (\Omega, \mathcal{M}_{n+1}(\R))$;
 \ \\
 \item (H2)
 $\displaystyle{
 \mbox{for all}\ (x,z) \in \Omega \ \ m(x,z)>0 \,   \ \mbox{and}\  M(x,z) = M(x,z)^t >0;
 }$
 \ \\
 \item (H3) $m$ and $M$ are integrable on every compact subset of $\bar \Omega$.
 \end{itemize}

 We define the Hilbert space
 $$
 H=\{\varphi \in L^2_{loc}(\O): \int_{\O}|\varphi(x,z)|^2 m(x,z) d\mu <\infty \},
 $$
 and the energy space
 $$
 \E = \{\varphi \in H \cap H^1_{\mbox{loc}}(\O): b(\varphi,\varphi)<\infty \},
 $$
where
\begin{equation}\label{bL}
b(\varphi,\psi)= \int_{\O} M(x,z)\, \mbox{grad}\, \varphi(x,z)
\cdot \mbox{grad}\, \psi(x,z)\, d\mu
\end{equation}
for suitable $\varphi, \psi$. They are equipped with their
canonical norms. Thanks to hypothesis (H3), $C_c^{\infty}(\O)$ is
included in $H$ and $\E$, and we moreover assume
 \begin{itemize}
 \item (H4) $C_c^{\infty}(\O)$ is dense in $H$ and in $\E$.
 \end{itemize}

 Then, the operator $L$ is defined on $C_0^{\infty}(\O)$ and it is symmetric by (H2).
So, we  ask when $L$ has the property we call {\em alternative
well-posedness},
 which, by definition, means that there is only one self-adjoint extension
 of $L$ with domain included in $\E$. A first answer to this question is the  following result.

\medskip

\begin{theorem}\label{UEP}
Let $L$ be a divergent operator fulfilling hypotheses (H1-H4)
\begin{enumerate}
\item Assume $L$ has the { alternative well-posedness property}.
Then, for every measurable and non negligible set $\Gamma$ in
$\R^n$, we have

\beq\label{N}
 \int_0^1 \int_{\Gamma}\frac{1}{m_{n+1,n+1}(x,z)}dx\, dz=\infty.\ \
 \eeq

\item Assume that for all $x \in \R^n$(everywhere, not almost
everywhere) there exists an open ball $B$ containing $x$ such that
 \beq\label{S}
 \int_0^1 \frac{1}{\omega_{B}(z)}dz=\infty,\ \
 \eeq
where $\omega_{B}(z)=\int_{B} m_{n+1,n+1}(y,z)dy$. Then $L$ has the {
alternative well-posedness property}.

 \end{enumerate}
\end{theorem}
 This theorem says that only the normal diagonal part of
$L$, namely the term $-\pd z m_{n+1,n+1} \pd z$, seems to matter
with respect to the{ alternative well-posedness} and that it
should vanish rapidly enough on approaching the boundary for this
property to hold.
\ \\
The necessary condition in ({\it i}) is not sufficient: there
exists an operator $L$ which does not have the {alternative
well-posedness}, but such that
\beq\label{N}
\int_0^1 \int_{\Gamma}\frac{1}{m_{n+1,n+1}(x,z)}dx\,dz=\infty\ \ \eeq
 for every
non negligible subset $\Gamma$ of $\R^n$.

Also, the sufficient condition in ({\it ii}) is not necessary: there
exists an operator $L$ which has the {alternative well-posedness},
but such that \beq\label{a}
 \int_0^1 \frac{1}{\omega_{B}(z)}dz<\infty\ \
 \eeq
 for all balls $B$ containing the origin.
However, this condition is sharp in the sense that there exists an
operator $L$ which has not the {alternative well-posedness
property}, satisfying (\ref{S}) for every $x\in \R^n$ but $0$.

 Nevertheless, an immediate and useful corollary, which
in particular applies when $ m_{n+1,n+1}$ only depends on the
variable $z$, is the following. Its proof is left to the reader.

\begin{cor}
Assume, in addition, that for all $x \in \R^n$ there exists $\rho
>0$ such that
\begin{equation}
\sup_{z\leq 1}\left( \int_{B(x,\rho)} m_{n+1,n+1}(y,z)dy
\right)\left( \int_{B(x,\rho)} \frac{1}{m_{n+1,n+1}(y,z)}dy
\right)< \infty.
\end{equation}
Then $L$ has the {alternative well-posedness property}
if and only if, for all balls $B$ in $\R^n$, we have
\begin{equation}
\int_0^1 \int_{B} \frac1{m_{n+1,n+1}(y,z)}dy \, dz= \infty.
\end{equation}
\end{cor}

\medskip
A remark is here in order. The reader has noticed that our
hypothesis (H1) on the regularity of the coefficients is obviously
not sharp. They are designed to avoid additional difficulties
which are not essential for our understanding of the alternative
well-posedness property.

\section{Proof of Theorem \ref{UEP}}

We begin with an abstract characterization of the {alternative
well-posedness property}. Let us denote by $\E_0$ the closure of
$C_0^{\infty}(\O)$ in $\E$. Then, we have

\begin{lem}
The operator $L$ has the alternative well-posedness property if
and only if $\E_0 = \E$.
\end{lem}

 To prove this assertion, we begin with assuming that $L$
has the alternative well-posedness property. Let $\mathcal{L}$ and
$\mathcal{L}_0$ be the self-adjoint operators associated with the
form $b$ given by (\ref{bL}) and defined on domains $\E$ and $\E_0$ respectively.
They are both extensions of $L$ with domains included in $\E$, and so, are equal.
But then we must have $D(\mathcal{L}^\frac12)=D(\mathcal{L}_0^\frac12)$, which is $\E = \E_0$.

Reciprocally, if $\E = \E_0$, the only self-adjoint
extension of $L$ with domain in $\E$ is its Friedrichs extension,
because the form $b$ defined on $\E$ is the closure of the form
$b$ defined on $C_0^{\infty}(\O)$. The lemma is proved.
\begin{flushright}
$\Box$
\end{flushright}

\begin{rem}
This Lemma, as simple as it is, enlightens the key point in the
alternative well-posedness property. Indeed, in a classical case
where boundary conditions are needed, the space $\E$ is the
largest possible Banach space onto which $(b(\varphi,\varphi) + \|
\varphi \|^2_H)^{\frac12}$ defines a norm, while $\E_0$ is the
smallest one containing $C^\infty_0 (\Omega)$. One thus may
understand the { alternative well-posedness} as saying that the
smallest possible space is also the largest, and this is why no
boundary condition is needed to obtain a self-adjoint extension
with domain in $\E$.
\end{rem}

 We turn to the proof of the first assertion of Theorem
\ref{UEP}. We assume the existence of a non negligible subset
$\Gamma$ of $\R^n$ such that
\begin{equation}\label{gamma}
 \int_0^1 \int_{\Gamma}\frac{1}{m_{n+1,n+1}(x,z)}\,dx\,dz<\infty.\ \
 \end{equation}

We will prove that $\E_0\neq \E$ by constructing a linear form on
$\E$ vanishing on $\E_0$, but not identically vanishing.

 Let $\eta\in C_c^{\infty}(0,\frac12)$ which equals
1 in a neighbourhood of $0$. If $\varphi \in  \E$ we set
$$
\lambda(\varphi)=\int_0^1\int_{\Gamma} \pd z
(\varphi(x,z)\eta(z))dx\, dz.
$$
That $\lambda$ defines a continuous linear form on $\E$ follows
directly from Cauchy-Schwarz inequality and (\ref{gamma}); it is
vanishing on $\E_0$ but not on $\E$, since
$\lambda(\varphi)=-\int_{\Gamma}\varphi(x,0)dx$ whenever $\varphi
\in C_c^{\infty}(\Omega)$.

\medskip
 Regarding the second assertion, we assume $\E \neq \E_0$
and prove that (\ref{S}) does not hold. There exists $\lambda \in
\E'$ (the dual space of $\E$), not identically vanishing, but null on $\E_0$. Therefore,
there is at least one (and in fact many, as we will see) test
function $\varphi \in C_c^{\infty}(\O)$ such that
$\lambda(\varphi)\neq 0$.

 The first step consists in showing that whenever
$\lambda(\varphi)\neq 0$ (suppose $\lambda(\varphi)=1$), then
$\int_0^1\frac1{\omega_B}<\infty$ as soon as $B\supset
\mbox{Supp}\ \phi(\cdot,0)$. For $x\in\R^{n}$, we write
$\zeta_0(x)=\varphi(x,0)$ and we define
$$
 E=\{\eta \in H^{1}_{loc}(0,\infty): \ \zeta_0
\otimes \eta \in \E \},
$$
equipped with the norm $\displaystyle{\|\eta\|_E=\|\zeta_0 \otimes \eta \|_{\E}} .$ It is
a Banach space, on which
$C_c^{\infty}(0,\infty)$ is dense.

Let $\eta \in C_c^{\infty}(0,\infty)$ and $\psi = \eta
(0)\varphi-\zeta_0 \otimes \eta$. Since $\psi$ vanishes on
$\partial \O$, it belongs to $\E_0$ (extend it by $0$ outside $\O$
and remark that $\psi_{\ve}(x,z) = \psi(x,z-\ve)$ converge towards
$\psi$ in $\E$). We thus have
$$
\lambda(\zeta_0 \otimes \eta )=\lambda(\eta(0)\varphi)=\eta(0)
$$
and
\begin{equation}\label{eta(0)}
|\eta(0)| \leq \| \lambda \|_{\E'} \| \eta \|_{E}.
\end{equation}

We compute
 $$ \|\eta\|_E^2=\int_0^{\infty}\omega(z)\eta '(z)^2\,
dz + \int_0^{\infty}\beta(z)\eta '(z)\eta(z)\, dz
+\int_0^{\infty}(\alpha (z)+c(z))\eta(z)^2\, dz $$
where, decomposing the matrix $M$ as
$$ M=\left(\begin{array}{ccccc}
 & & & & m_{1,n+1}\\
  & & & &.\\
 & & M'& & .\\
 & & & &.\\
 m_{1,n+1}& .& .& .& m_{n+1,n+1}
\end{array}\right),
$$
we have set
\begin{eqnarray*}
\omega(z)&=&\int_{\R^{n}} m_{n+1,n+1}(x,z) \zeta_0(x)^2\, dx,\\
\beta(z)&=&2\sum_{j=1}^{n}\int_{\R^{n}} m_{j,n+1}(x,z) \partial_j\zeta_0(x)\zeta_0(x)\, dx,\\
\alpha(z)&=&\int_{\R^{n}} M'(x,z) \nabla\zeta_0(x) \cdot \nabla\zeta_0(x)\, dx,\\
c(z)&=&\int_{\R^{n}} \zeta_0(x)^2\,m(x,z)\, dx.\\
\end{eqnarray*}

By the positive definiteness of $M$, we have for all $z>0$
$$
 \beta(z)\leq 2 \sqrt{\a(z)\omega(z)}.
$$
Inserting this in (\ref{eta(0)}), we find that
$$
|\eta(0)|^2 \leq C \left(\int_0^{\infty}w(z)\eta '(z)^2\, dz +
\int_0^{\infty}(\alpha (z)+c(z))\eta(z)^2\, dz\right)
$$
for some uniform constant $C$. It is not difficult to see that this
last inequality implies
$$ \int_0^{1}\frac{1}{w(z)}\, dz \leq C,
$$
and then \beq \label{wB} \int_0^{1}\frac{1}{\omega_B(z)}\,
dz <\infty , \eeq for all balls $B$ which contain the support of
$\zeta_0$.

The second step is an argument of descent which will give one
point $x$ in $\R^n$ such that (\ref{wB}) holds when $B$ contains
$x$. We will use for each $ j \geq 0 $ a partition of unity
$(\chi_{j, k})_{k\in \mathbb{Z}^n}$with $\mbox{Supp}\ \chi_{j, k}
\subset B(k 2^{-j}, \sqrt{n} 2^{-j})=:B_{j,k}$.

 We start with $j=0$, and decompose $\varphi= \sum_k
\varphi\chi_{0,k}$. There must exist at least one $k_0 \in
\mathbb{Z}^n$ such that $\lambda (\varphi\chi_{0,k_0})\neq 0$. By
the first step, the inequality (\ref{wB}) is verified for
$B=B_{0,k_0}$.

 We now decompose $\varphi_0
=\varphi\chi_{0,k_0}$ at the next finer scale: $\varphi_{0}=\sum_k
\varphi_0\chi_{1,k}$, and select $k_1$ such that $\lambda
(\varphi_0\chi_{1,k_1})\neq 0$. Then, (\ref{wB}) holds for
$B=B_{1,k_1}$ (note that this ball is included in $B_{0,k_0}$).

We continue, and thus inductively construct a nested sequence of
dyadic balls fulfilling (\ref{wB}): the point $x$ we are looking
for is at the intersection of these balls. Theorem \ref{UEP} is
completely proved.

\section{A few conclusive words}

Going back to spacetimes with naked singularities, the alternative
well posedness property being satisfied means that such
spacetimes, though exhibiting a boundary, are physically
consistent because their Laplace-Beltrami operator is well defined
without requiring any condition at the boundary. We are not
claiming that naked singularities do exist in the universe, which
we do not know. But, at least, the necessity to define a condition
at the boundary, fortunately, does not exist.

 However we have proved the {alternative well-posedness
property} only for the examples of this paper. If we believe in
General Relativity, then it is conceivable that any meaningful
solution having naked singularities should fulfill the {
alternative well-posedness property}.
\begin{appendix}
\section{Proof of Proposition \ref{vertical}}\label{apendice}

 Since the characteristics of the  equation \eq{fe} are
$z\pm t=$ constant, the change of variables $\xi=z+t$ and
$\eta=z-t$ brings this equation to $L(\phi)=0$, where
$$ L(\phi)=\pd{\xi \eta} \phi+
\frac{1}{2(\xi+\eta)}(\pd \xi \phi+ \pd \eta \phi ).
$$

 The adjoint  $L^*$ of the  operator $L$
is given by \beqa\lab{Lad} L^*(\psi)=\pd{\xi\eta} \psi-\pd
\xi\left(\frac{\psi}{2(\xi+\eta)}\right)- \pd \eta
\left(\frac{\psi}{2(\xi+\eta)}\right)  \eeqa and we have
\beqa
  \psi L(\phi)-\phi L^*(\psi)=-\pd\xi\left(\phi\; \pd
\eta\psi-\frac{\phi\, \psi}{2(\xi+\eta)}\right) + \pd\eta\left(\pd
\xi \phi\; \psi+\frac{\phi\, \psi}{2(\xi+\eta)}\right). \nonumber
\eeqa

Taking $\phi$ such that $L(\phi)=0 $ and integrating over a
piecewise smooth compact domain $\Gamma$ with boundary
$\partial\Gamma$, one obtains by Green's formula \beqa \lab{circ}\fl
-\int\int_\Gamma \phi L^*(\psi)\, d\xi\,d\eta=
\oint_{\partial\Gamma}\left(\pd \xi \phi\; \psi\,d\xi+\phi\,\pd
\eta  \psi\,d\eta +\frac{\phi\,
\psi}{2(\xi+\eta)}(d\xi-d\eta)\right)\,,\eeqa where the line
integral around $\partial\Gamma$ is taken in the  clockwise sense.

 To obtain a  representation for
$\phi(P)=\phi(\xi_0,\eta_0)$, following Riemann (see, for example,
\cite{CHII}), we choose for $\psi$ a function
$R(\xi,\eta;\xi_0,\eta_0)$ subject to the following conditions:

a) As a function of $\xi$ and $\eta$, $R$ satisfies the adjoint
equation
$$ L_{(\xi,\eta)}^*(R)=0\,.
$$

b) On the characteristics $\eta=\eta_0$ and $\xi=\xi_0$ it
satisfies
$$ \pd \xi R-\frac{R}{2(\xi+\eta)}=0\hspace{2cm}\text{on}\hspace{.5cm} \eta=\eta_0\,,
$$ and
$$ \pd \eta  R-\frac{R}{2(\xi+\eta)}=0\hspace{2cm}\text{on}\hspace{.5cm} \xi=\xi_0\,.
$$

c) $R(\xi_0,\eta_0;\xi_0,\eta_0)=1$.

 Conditions b) are ordinary differential  equations along
the characteristics; integrating them  and using c) we get that
$$
R(\xi,\eta;\xi_0,\eta_0)=\sqrt{\frac{\xi+\eta}{\xi_0+\eta_0}}\,.
$$
on both characteristics.

 Now by trying with the ansatz
$$
R(\xi,\eta;\xi_0,\eta_0)=\sqrt{\frac{\xi+\eta}{\xi_0+\eta_0}}\,F(w)\,,
$$ where

\beq\lab{w} w=-\frac{(\xi-\xi_0)(\eta-\eta_0)}
{(\xi_0+\eta_0)(\xi+\eta)}\,,
\eeq
and using \eq{Lad} we get that in order to satisfy condition a)
the function $F(w)$ must satisfy the differential equation
\beq\lab{hde} w(1-w)F''(w) + (1-2w)F'(w)-\frac{1 }{4}F(w)=0\,.\eeq
The only solution of this equation with $F(0)=1$ is the Gauss
hypergeometric function $\;_2F_1\,(\frac{1}{2},\frac{1}{2};{1};w)$
\cite{abra}. Therefore we have that the Riemann function is
\beq\lab{R}
R(\xi,\eta;\xi_0,\eta_0)=\sqrt{\frac{\xi+\eta}{\xi_0+\eta_0}}\,\;
_2F_1\!\left(\frac{1}{2},\frac{1}{2}
;{1};-\frac{(\xi-\xi_0)(\eta-\eta_0)}{(\xi_0+\eta_0)(\xi+\eta)}\right)\,,
\eeq which is a $C^{\infty}$ function of $\xi$ and $\eta$, for
$\xi+\eta>0$, $\xi_0+\eta_0>0$ and $|w|<1$.

 It can be
shown that as a function of $\xi_0$  and $\eta_0$, $R$ also
satisfies the  equation
$$
 L_{(\xi_0,\eta_0)}(R)=0\,.
$$
For the sake of clearness,  we introduce new time-like coordinates
$t=(\xi_0-\eta_0)/2$ and $ \tau=(\xi-\eta)/2$ and spacelike ones $
z=(\xi_0+\eta_0)/2$ and $\zeta=(\xi+\eta)/2$. In terms of these
coordinates, we have
$$ w=\frac{(t-\tau)^2-(z-\zeta)^2}{4 z \zeta}\,,
$$
and
 \beq \lab{Rtz}
R(\tau,\zeta;t,z)=\sqrt{\frac{\zeta}{z}}\,\;
_2F_1\!\left(\frac{1}{2},\frac{1}{2};{1};\frac{(t-\tau)^2-(z-\zeta)^2}{4
z \zeta}\right)\,. \eeq

\begin{figure}[h!]\begin{center}
\includegraphics[width=\textwidth]{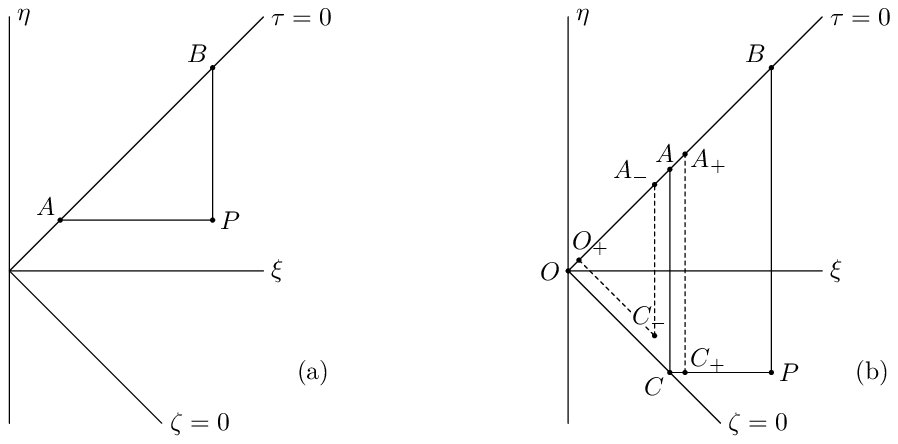}\caption{Domain of integration $\Gamma$: (a) for $z>t$ and (b) for $z<t$.}{\lab{ABC}}\end{center}
\end{figure}

 It can readily  be seen from \eq{w}  that, in the
region $0\leq\xi\leq\xi_0$, $ \eta\geq\eta_0$,  $\xi+\eta>0$ and
$\xi_0+\eta_0>0$  (see Fig. \ref{ABC}), it holds that $w\geq0$. If
in addition $\eta_0>0$ ($z>t$), we can easily check that $0\leq
w<1$  and then  Riemann's function \eq{R} is well-defined in that
region. However, if $\eta_0<0$ ($z<t$), we see that $w=1$ on the
characteristic $\xi=-\eta_0$, and furthermore $w>1$  for
$0\leq\xi<-\eta_0$. Thus, in the latter case,
$R(\xi,\eta;\xi_0,\eta_0)$ is defined through \eq{R} only for
$-\eta_0<\xi\leq\xi_0$.

 Therefore we must treat the cases $z>t$ and $z<t$
separately.

\subsection{$z>t$}

 In this case, since $R(\xi,\eta;\xi_0,\eta_0)$ is an
infinitely differentiable solution of $L_{(\xi,\eta)}^*(\psi)=0$
in $0\leq\xi\leq\xi_0$ and $ \eta\geq\eta_0$, by applying
expression \eq{circ} to the triangle $ABP$  (see Fig.\ref{ABC}(a))
we get
 \beqa
\int_{AB}\left(\pd \xi \phi\; R\,d\xi+\phi\,\pd \eta
R\,d\eta+\frac{\phi\,R}{2(\xi+\eta)}(d\xi-d\eta)\right)&&\nn  +
\int_{BP}\phi\,\left(\pd \eta  R-\frac{
R}{2(\xi+\eta)}\right)d\eta +\int_{PA}\left(\pd \xi \phi\;
R+\frac{\phi\, R}{2(\xi+\eta)}\right)d\xi &=&0 \,,\nonumber \eeqa
since $d\xi=d\eta$ along $AB$  and
$$
\int_{PA} \pd \xi \phi\,R\, d\xi= \phi(A) R(A)-\phi(P)
R(P)-\int_{PA}  \phi\,\pd \xi R\, d\xi\,,
$$
we have
 \beqa    \int_{AB}\left(\pd \xi \phi\;
R\,d\xi+\phi\,\pd \eta  R\,d\eta\right)+\int_{BP}\phi\,\left(\pd
\eta  R-\frac{ R}{2(\xi+\eta)}\right)d\eta &&\nn + \phi(A)
R(A)-\phi(P) R(P)-\int_{PA}\phi\;\left(\pd \xi  R -\frac{
R}{2(\xi+\eta)}\right)d\xi &=&0 \,.\nonumber \eeqa

Taking into
account that $R$ satisfies conditions b) and c), we get  Riemann's
representation formula \beqa    \phi(P)=\phi(A)
R(A)+\int_{AB}\left(\pd \xi \phi\; R\,d\xi+\phi\,\pd \eta
R\,d\eta\right) \,.\nonumber \eeqa

 In order to obtain a more
symmetric expression, we add the identity \beqa
0=\frac{1}{2}\bigl(\phi(B) R(B)-\phi(A) R(A)\bigr)-\frac{1}{2}
\int_{AB}\bigl(\pd \xi (\phi R)\,d\xi+\pd \eta(\phi
R)\,d\eta\bigr) \,\nonumber \eeqa getting \beqa
\phi(P)&=&\frac{1}{2}\bigl(\phi(B) R(B)+\phi(A) R(A)\bigr)\nn &+&
\frac{1}{2}\int_{AB}\bigl(\pd \xi \phi\; R-\phi\;\pd \xi
R\,\bigr)d\xi-\frac{1}{2}\int_{AB}\bigl(\pd \eta\phi\, R-\phi\,\pd
\eta  R\bigr)\,d\eta \,.\nonumber \eeqa
 Finally, by using
\eq{Rtz} we get (\ref{z>t}).

\subsection{$z<t$}

 In this case,  $R(\xi,\eta;\xi_0,\eta_0)$ is an
infinitely differentiable solution of $L_{(\xi,\eta)}^*(\psi)=0$
in $-\eta_0<\xi\leq\xi_0$ and  $ \eta\geq\eta_0$. For $\epsilon>0$
and sufficiently small, let us consider the trapezium
$A_+BPC_+A_+$ with vertices
$A_+=(-\eta_0+\epsilon,-\eta_0+\epsilon)$,  $B=(\xi_0,\xi_0)$,
$P=(\xi_0,\eta_0)$ and  $C_+=(-\eta_0+\epsilon,\eta_0)$ (see
Fig.\ref{ABC}(b)), by using expression \eq{circ}   we have \beqa
0= \oint_{A_+BPC_+A_+}\left(\pd \xi \phi\; R\,d\xi+\phi\,\pd \eta
R\,d\eta+\frac{\phi\,R}{2(\xi+\eta)}(d\xi-d\eta)\right)\,,\nonumber
\eeqa and proceeding as in the $z>t$ case we get \beqa
\phi(P)&=&\phi(C_+) R(C_+)+\int_{A_+B}\left(\pd \xi \phi\;
R\,d\xi+\phi\,\pd \eta R\,d\eta\right)\nn &+&
\int_{C_+A_+}\phi\,\left(\pd \eta  R-\frac{
R}{2(\xi+\eta)}\right)d\eta \,.\nonumber \eeqa

 Now, taking into
account that \beq\lab{log}\fl
_2F_1\,\left(\frac{1}{2},\frac{1}{2};{1};1-\theta\right)=-\frac{1}{\pi
}\log\left(\frac{\theta}{16}\right)-\frac{1}{4
\pi}\left(\log\left(\frac{\theta}{16}\right)+2\right)
\theta+O\left(\theta^2\,\log\theta\right), \eeq a straightforward
computation from \eq{R} shows that, on the characteristic
$C_+A_+$, {\it i.e}., $\xi=-\eta_0+\epsilon$, \beqa \left(\pd \eta
R-\frac{ R}{2(\xi+\eta)}\right)\Bigg
\vert_{\xi=-\eta_0+\epsilon}=\frac{1}{\pi(\xi_0+\eta)}
\sqrt{\frac{\xi_0+\eta_0}{\eta-\eta_0}}+O(\epsilon\,
\ln\epsilon)\,.\nonumber \eeqa Therefore we have \beqa
\phi(P)&=&\phi(C_+)
\sqrt{\frac{\epsilon}{\xi_0+\eta_0}}\,+\int_{A_+B}\left(\pd \xi
\phi\; R\,d\xi+\phi\,\pd \eta  R\,d\eta\right)\nn &+&
\frac{1}{\pi}
\int_{C_+A_+}\frac{\phi}{(\xi_0+\eta)}\sqrt{\frac{\xi_0+\eta_0}
{\eta-\eta_0}}\;d\eta +O(\epsilon\,\ln\epsilon) \,,\nonumber \eeqa
 and when $\epsilon$ goes to $0$, we get
\beqa   \lab{fiP}\fl \phi(P)=\int_{AB}\left(\pd \xi \phi\;
R\,d\xi+\phi\,\pd \eta
R\,d\eta\right)+\frac{1}{\pi}\int_{\eta_0}^{-\eta_0}
\frac{\phi(-\eta_0,\eta)}{(\xi_0+\eta)}
\sqrt{\frac{\xi_0+\eta_0}{\eta-\eta_0}}\;d\eta \,. \eeqa

In order
to get rid of the integral along the  characteristic $CA$ in the
last expression we shall ``extend" the Riemann's function $R$ to
the region $w>1$.

 Let us consider the function
 \beqa \lab{R1}
R_1(\xi,\eta;\xi_0,\eta_0)=\sqrt{\frac{\xi+\eta}{\xi_0+\eta_0}}
\,\frac{1}{\sqrt{w}}\;
_2F_1\!\left(\frac{1}{2},\frac{1}{2};{1};\frac{1}{w}\right)\nn =
\frac{\xi+\eta}{\sqrt{(\xi_0-\xi)(\eta-\eta_0)}}\,\;
_2F_1\!\left(\frac{1}{2},\frac{1}{2};{1};-\frac{(\xi_0+\eta_0)
(\xi+\eta)}{(\xi-\xi_0) (\eta-\eta_0)}\right)\,. \eeqa Clearly
$R_1$ vanishes at the boundary ($\xi+\eta=0$) and by comparing
with \eq{R}, we see that $R$ and $R_1$ ``formally coincide" at
$w=1$, although, of course,  none of them exist there (see
equation \eq{log}).

On the other hand, since
$\displaystyle{\frac{1}{\sqrt{w}}\;
_2F_1\!\left(\frac{1}{2},\frac{1}{2};{1};\frac{1}{w}\right)}$  is
the regular solution at $\infty$ of the hypergeometric
differential equation \eq{hde}, $R_1$
 also satisfies
$$
 L_{(\xi,\eta)}^*(R_1)=0 \hspace{1cm}\text{and}
\hspace{1cm}L_{(\xi_0,\eta_0)}(R_1)=0 \,,
$$
as it can be readily checked. Moreover straightforward
computations show that in a neighborhood of the boundary \beqa
\lab{r1} \left(\pd \eta  R_1-\frac{ R_1}{(\xi+\eta)}\right)\Bigg
\vert_{\xi+\eta=\epsilon}=-\frac{\bigl(\xi_0-\xi+\eta-\eta_0\bigr)
}{4 \bigl((\xi_0-\xi)(\eta-\eta_0)\bigr)^{3/2}}\;\epsilon
+O\left(\epsilon^2\right)\, \eeqa and that in a neighborhood of
the characteristic $CA$ \beqa  \lab{r2} \left(\pd \eta  R_1-\frac{
R_1}{2(\xi+\eta)}\right)\Bigg
\vert_{\xi=-\eta_0-\epsilon'}=\frac{1}{\pi(\xi_0+\eta)}
\sqrt{\frac{\xi_0+\eta_0}{\eta-\eta_0}}+O(\epsilon'\,\ln\epsilon')\,.
\eeqa

For $\epsilon$ and $\epsilon'$ positive  and small enough,
let us consider the triangle $O_+A_-C_-$ with vertices
$O_+=(\epsilon,\epsilon)$,
$A_-=(-\eta_0-\epsilon',-\eta_0-\epsilon')$ and
$C_-=(-\eta_0-\epsilon',\eta_0+2\epsilon+\epsilon')$ (see
Fig.\ref{ABC}(b)). Since $R_1(\xi,\eta;\xi_0,\eta_0)$ is an
infinitely differentiable solution of $L_{(\xi,\eta)}^*(\psi)=0$
in $\xi+\eta>0$ and $0\leq\xi<-\eta_0$, by using expression
\eq{circ} we have \beqa \int_{O_+A_-}\left(\pd \xi \phi\;
R_1\,d\xi+\phi\,\pd \eta R_1\,d\eta
\right)+\int_{A_-C_-}\phi\,\left(\pd \eta R_1-\frac{.
R_1}{2(\xi+\eta)}\right)d\eta &&\nn +\int_{C_-O_+}\pd \xi \phi\;
R_1\,d\xi+\int_{C_-O_+}\phi\,\left(\pd \eta  R_1-\frac{
R_1}{(\xi+\eta)}\right)d\eta =0 \,.\nonumber \eeqa Taking into
account \eq{R1}, \eq{r1} and \eq{r2} we can write \beqa \fl
\int_{O_+A_-}\hspace{-.5cm}\left(\pd \xi \phi\;
R_1\,d\xi+\phi\,\pd \eta  R_1\,d\eta \right)+\frac{1}{\pi}
\int_{A_-C_-}\frac{\phi}{(\xi_0+\eta)}
\sqrt{\frac{\xi_0+\eta_0}{\eta-\eta_0}}\;d\eta
+O(\epsilon,\epsilon'\,\ln\epsilon')=0 \,,\nonumber \eeqa and when
$\epsilon$ and $\epsilon'$ go to $0$, we get \beqa
\int_{OA}\left(\pd \xi \phi\; R_1\,d\xi+\phi\,\pd \eta  R_1\,d\eta
\right)-\frac{1}{\pi}\int_{\eta_0}^{-\eta_0}\frac{\phi(-\eta_0,\eta)}{(\xi_0+\eta)}
\sqrt{\frac{\xi_0+\eta_0}{\eta-\eta_0}}\;d\eta =0\,.\nonumber
\eeqa Therefore, by adding this expression to \eq{fiP}, we get
Riemann's representation formula for the case $\eta_0<0$ \beqa
\phi(P)=\int_{OA}\left(\pd \xi \phi\; R_1\,d\xi+\phi\,\pd \eta
R_1\,d\eta \right)+\int_{AB}\left(\pd \xi \phi\; R\,d\xi+\phi\,\pd
\eta  R\,d\eta\right)\,.\nonumber \eeqa
 Finally, by using
\eq{R1} we obtain(\ref{z<t}).

 Notice that, since  no contribution from the boundary
($\xi+\eta=0$) remains in the last expression, the solution is
completely determined by the initial data. Therefore no boundary
condition should and can be provided. However, at the boundary,
(\ref{z<t}) becames
\beqa\fl  \phi(t,0)&=&
   \int_{0}^{t} \Bigl( f'(\zeta) +g(\zeta)\Bigr) \frac{\zeta }{\sqrt{t^2
   -\zeta^2}} \, d\zeta
  +\,t \int_{0}^{t}\, \frac{f(\zeta) }{(t+\zeta)\;\sqrt{t^2
   -\zeta^2}} \, d\zeta\,,
 \eeqa
 and integrating  by parts we get the trace given in (\ref{z=0}).

 Notice that, for $z=t$, we can easily see that expressions
\eq{z>t} and \eq{z<t} coincide.
\begin{flushright}
$\Box$
\end{flushright}

\end{appendix}

\end{document}